\documentclass[11pt]{article}

\usepackage{amsfonts}
\usepackage{amsmath, amsthm, amsfonts, amssymb,mathrsfs,color}
\usepackage{authblk} 
\usepackage{mathtools}
\usepackage{charter}
\usepackage{orcidlink}
\usepackage{CJK}
\usepackage{tensor}

\usepackage{url}
\usepackage{todonotes}

\usepackage{soul}
\usepackage{enumitem}
\usepackage[normalem]{ulem}
\usepackage[numbers,sort&compress]{natbib}

\usepackage
[
a4paper,
left=2cm,right=2cm,top=2.5cm,bottom=2.1cm
]{geometry}

\usepackage{indentfirst}
\usepackage{extarrows}
\usepackage{graphicx}
\usepackage{float}
\usepackage{multicol}
\usepackage{multirow}
\usepackage{booktabs}
\usepackage{longtable}
\usepackage{bm}
\usepackage{appendix}
\usepackage{thmtools}
\usepackage{setspace}

\usepackage[caption=false]{subfig}
\usepackage[]{caption2} 
 
\usepackage{pdflscape}
\usepackage[figuresright]{rotating}

\usepackage{hyperref}
\hypersetup{
	colorlinks,
	citecolor=blue,
	filecolor=black,
	linkcolor=blue,
	urlcolor=blue,
	hypertexnames=false
}

\theoremstyle{plain}
\numberwithin{equation}{section}
\newtheorem{Theorem}{Theorem}[section]

\newtheorem{Corollary}{Corollary}[section]

\newtheorem{Definition}{Definition}[section]

\theoremstyle{definition}

\newtheorem{Remark}{Remark}

\newtheorem{Assumption}{Assumption}[section]

\usepackage{fancyhdr}
\title{{\bf Optimal Investment and Consumption Strategies with General Cost Structure under CRRA Utility}}

\author[a]{\textsc{Yingting Miao} \thanks{Email: \href{Yingting.Miao@xjtlu.edu.cn}{Yingting.Miao@xjtlu.edu.cn}}}
\author[b]{\textsc{Qiang Zhang} \thanks{Email: \href{mazq@uic.edu.cn}{mazq@uic.edu.cn} (Corresponding author)}}

\affil[a]{{\small Department of Financial and Actuarial Mathematics, Xi'an Jiaotong-Liverpool University, Suzhou, Jiangsu 215123, P. R. China}}

\affil[b]{{\small Research Center of Mathematics, Advanced Institute of Natural Sciences, Beijing Normal University, Zhuhai 519087, P. R. China;\\ Guangdong Provincial Key Laboratory of Interdisciplinary Research and Application for Data Science, Beijing Normal-Hong Kong Baptist University, Zhuhai 519087, P. R. China}}

\begin{document}
	\begin{CJK}{UTF8}{gbsn} 
\date{\today}
\maketitle
\begin{abstract}
	Transaction costs play a critical role in portfolio allocation and consumption decisions. We study a finite-horizon consumption--investment problem with CRRA utility under a general class of transaction cost functions. Based on dynamic programming and a singular perturbation expansion for a small cost-to-wealth ratio, we derive leading-order asymptotic formulas for the no-trade region, the four trading boundaries, the value function correction, and the optimal consumption rate. We further show how fixed, proportional, fixed-plus-proportional, and nonlinear transaction costs arise as special cases of the general framework. The results show that the leading-order no-trade region is governed by the fixed and proportional components, while the framework still accommodates nonlinear cost structures. 
	Complementing the asymptotic analysis, we prove a verification theorem for the exact impulse-control formulation under a strictly positive fixed cost component, and characterize its limiting transitions to singular and continuous control regimes as the fixed cost vanishes.

\end{abstract}


\section{Introduction}\label{intro}
Investment and consumption decisions are central to portfolio management, with significant practical implications for fund managers. Merton \cite{Merton-RES-1969, Merton-JET-1971} conducted pioneering work in the field of continuous-time portfolio selection. Nevertheless, Merton's framework requires fund managers to adjust asset allocations continuously over time, a task often rendered impractical in the presence of transaction costs, as it can quickly deplete their wealth. These costs substantially affect portfolio choices for both individual and institutional investors. Nonetheless, optimal investment problems that involve transaction costs are widely acknowledged for their inherent complexity and difficulty. Building upon Merton's foundational work, Magill and Constantinides \cite{Magill-Constantinides-JET-1976} introduced transaction costs into the original model, demonstrating the existence of a no-trade region.

Consequently, the fund manager will buy a certain number of shares of the risky asset when the risky allocation reaches the buy boundary and will sell a certain number of shares of the risky asset when the risky allocation reaches the sell boundary. The critical questions are: \textit{When should the investor buy the risky asset? How much should the investor buy? When should the investor sell the risky asset? How much should the investor sell?} Mathematically, these decisions are characterized by four free boundaries, which is a central source of analytical difficulty. A further question is \textit{what the optimal consumption strategy should be}. In this paper, we address the above questions and determine an optimal consumption strategy and an optimal asset allocation strategy for a general transaction cost structure, based on the maximization of the expected CRRA utility function at a finite investment horizon.

Motivated by the seminal work of Magill and Constantinides \cite{Magill-Constantinides-JET-1976}, a substantial literature has developed on optimal investment and consumption strategies with transaction costs. Transaction costs are of significant importance not only in portfolio optimization but also in various other areas of modern finance, including option pricing (Hodges et al.\ \cite{Hodges-Neuberger-RFuMa-1989}, Davis et al.\ \cite{Davis-Panas-Zariphopoulou-SIAMJCO-1993}), shadow prices (Kallsen and Muhle-Karbe \cite{Kallsen-Muhle-AAP-2010}), risk-sensitive asset management (Bielecki and Pliska \cite{Bielecki-Pliska-FS-2000}, Bielecki et al.\ \cite{Bielecki-Chancelier-Pliska-Sulem-JCF-2004}), and asset pricing (Lo et al.\ \cite{Lo-Mamaysky-Wang-JPE-2004}, Herdegen et al.\ \cite{Herdegen-Martin-etal-FS-2021}). The profound impact of transaction costs across finance has motivated a broad range of analytical approaches, including martingale techniques (Cvitanic and Karatzas \cite{Cvitanic-Karatzas-MF-1996}), numerical methodologies (Gonon et al.\ \cite{Gonon-Muhle-Karbe-Shi-MF-2021}), and asymptotic analysis methods \cite{Atkinson-Wilmott-MF-1995, Whalley-Wilmott-MF-1997, Korn-FS-1998}.

Transaction cost functions can be decomposed into fixed and proportional leading components, together with a nonlinear residual component.

\textbf{Proportional transaction costs:} Jane\v{c}ek and Shreve \cite{Janecek-Shreve-FS-2004} studied optimal investment and consumption strategies over an infinite investment horizon based on CRRA utility. Kallsen and Muhle-Karbe \cite{Kallsen-Muhle-MF-2017} investigated the optimal investment and consumption problem for a general utility function and obtained the explicit leading-order solution in an asymptotic expansion. Based on CRRA utility, Quek and Atkinson \cite{Quek-Atkinson-AMF-2017} considered a multi-period discrete-time setting in which the coefficient of the proportional transaction costs varies in each period. Chellathurai and Draviam \cite{Chellathurai-Draviam-MPES-2005} studied the optimal investment strategy in a setting where the coefficient of proportional transaction costs depends on trading volume. Yang \cite{Yang-book-2006} provided an explicit asymptotic solution for a finite investment horizon based on the exponential utility function. Constantinides examined the capital market with transaction costs; see \cite{Constantinides-JPE-1986}. Additional work can be found in \cite{Akian-Menaldi-Sulem-SIAMJCO-1996,Davis-Norman-MOR-1990,Liu-Loewenstein-RFS-2002,Mokkhavesa-Atkinson-IMAJMM-2002,Shreve-Soner-AAP-1994, Akian-Sulem-Taksar-MF-2001,Dumas-Luciano-JF-1991,Taksar-Klass-Assaf-MOR-1988}.

\textbf{Fixed transaction cost only:} Morton and Pliska \cite{Morton-Pliska-MF-1995} derived the optimal consumption strategy for an infinite investment horizon based on logarithmic utility and fixed transaction cost only. Altarovici et al. \cite{Altarovici-Muhle-Soner-FS-2015} considered independent multi-assets under CRRA utility over an infinite time horizon. Lo et al. \cite{Lo-Mamaysky-Wang-JPE-2004} proposed a dynamic equilibrium model that includes asset prices, trading volume, and fixed transaction costs.

\textbf{Both fixed and proportional costs:} Based on the exponential utility function, Liu \cite{Liu-JF-2004} formulated the governing dynamic equations and free boundary conditions for multiple assets, demonstrating that these equations reduce to the one-risky-asset problem when all assets are independent. He also numerically studied the case of two correlated risky assets. The Hamilton-Jacobi-Bellman quasi-variational inequality method has also been applied to study optimal consumption strategies over an infinite investment horizon \cite{Altarovici-Reppen-Soner-SIAMJCO-2017,Cadenillas-MMOR-2000,Korn-FS-1998,Oksendal-Sulem-SIAMJCO-2002}. Altarovici et al. \cite{Altarovici-Reppen-Soner-SIAMJCO-2017} and Cadenillas \cite{Cadenillas-MMOR-2000} considered general utility and CRRA utility of consumption, respectively. Korn \cite{Korn-FS-1998} studied general utility of consumption and derived the optimal solution for exponential utility maximization as an example. {\O}ksendal and Sulem \cite{Oksendal-Sulem-SIAMJCO-2002} examined the CRRA utility of consumption and presented numerical estimates for the value function and the optimal consumption strategy.

\textbf{Nonlinear transaction costs:} Empirical studies suggest that market frictions can be more complex. For example, Lillo et al. \cite{Lillo-Farmer-Mantegna-Nature-2003} show that the price impact of individual trades is well described by a smooth, concave function of trade size and can be rescaled across market-capitalization classes. Similarly, Almgren et al. \cite{Almgren-Thum-Risk-2005} decompose market impact into permanent and temporary components using institutional trading data and find evidence that the total transaction cost follows an $8/5$ power law; the corresponding price impact follows a $3/5$ power law. These findings provide robust evidence that the total transaction cost function $k(\Delta)$ for an investor is smooth and convex, thereby providing practical support for the smoothness and convexity assumptions adopted in our model (see Assumption~\ref{ass:cost} later). This motivates theoretical frameworks capable of accommodating the general nonlinear cost structures observed in financial markets.

	Motivated by these empirical observations, recent work has studied optimal strategies under nonlinear price impact. G{\^a}rleanu and Pedersen \cite{Garleanu-Pedersen-JF-2013} derive a closed-form dynamic portfolio policy under quadratic trading costs in a model with return predictability. 
	Moreau, Muhle-Karbe, and Soner \cite{Moreau-Muhle-Karbe-Soner-MF-2017} obtain asymptotically optimal policies and welfare losses for small linear price impact, while Cay\'{e} et al. \cite{Caye-Herdegen-Muhle-Karbe-AAP-2020} and Guasoni and Weber \cite{Guasoni-Weber-MF-2020} analyze nonlinear price impact with power-law trading-rate costs. These studies show that optimal trading rates and welfare losses depend on risk tolerance, market volatility, the volatility of the frictionless target strategy, and the elasticity of the impact function. Together, they point to a broader need for developing a more unified and flexible analytical framework.

There is growing interest in how more realistic transaction cost schedules affect optimal trading. For instance, Belak et al.\ \cite{Belak-Mich-Seifried-MF-2022} study retail-investor portfolios under fixed, fixed-plus-proportional, piecewise constant, and floored-and-capped proportional costs. Their numerical results show that the cost structure can materially affect the shape of the no-trade region and the associated rebalancing behavior. These findings further highlight the value of analytical frameworks that can accommodate broader transaction cost specifications.

Asymptotic methods are also useful in related portfolio problems. For example, in a model incorporating return predictability and learning, Wang and Siu \cite{Wang-Siu-EJOR-2024} use an asymptotic expansion technique to study the impact of small proportional transaction costs on optimal investment and consumption, quantifying the interplay between transaction costs, risk aversion, and signal uncertainty. Melnyk et al. \cite{Melnyk-Muhle-Karbe-etal-MF-2020} have shown that the leading-order solution for proportional transaction costs remains consistent for agents with both additive and recursive utilities. Melnyk and Seifried \cite{Melnyk-Seifried-MF-2018}, employing asymptotic analysis, investigate long-term growth rates under both proportional and Morton-Pliska transaction costs. Additionally, Chen et al. \cite{Chen-Dai-Jing-Qin-MF-2022} study the long-term portfolio choice problem involving two illiquid and correlated assets under proportional costs. Further work on transaction costs can be found in comprehensive surveys \cite{Soner-Touzi-SiamCO-2013, Possamai-Soner-Touzi-CPDE-2015}.

Despite these important contributions, further theoretical analysis of finite-horizon CRRA consumption -- investment problems under general transaction costs is still needed. Existing asymptotic results for nonlinear price impact typically focus on particular power-law forms with specific exponents, while richer cost schedules are often studied numerically. A unified leading-order analysis for a broad class of transaction cost functions is therefore useful.

It is useful to clarify the asymptotic regime considered in this paper. Our perturbation analysis is not based on the assumption that the dollar transaction cost \(k(\Delta)\) itself is small. Rather, in the CRRA setting, the relevant small quantity is the transaction cost relative to total wealth, \(k(\Delta)/w\). Thus, the no-trade region is small in terms of portfolio proportions, while the dollar amount traded, and hence the dollar transaction cost, need not be small in absolute terms.

This distinction reflects the homogeneity of CRRA preferences. Under exponential utility, the wealth variable can often be separated from the value function, and trading decisions are naturally expressed in terms of dollar amounts. A small-cost expansion in that setting is therefore naturally tied to small dollar trades, so the local behavior of the transaction cost function near the origin is the relevant object. By contrast, under CRRA utility, the optimal investment decision is expressed through the ratio of risky wealth to total wealth. Consequently, the perturbation parameter in our model is the wealth-normalized cost \(k(\Delta)/w\), not the absolute cost \(k(\Delta)\). The analysis should therefore be interpreted as a small no-trade-region expansion governed by the cost-to-wealth ratio.

We address this need by developing a unified asymptotic framework for finite-horizon CRRA investment and consumption problems with general transaction costs. We consider a general class of transaction cost functions and derive explicit leading-order expressions for the four free boundaries: the buy, post-buy, sell, and post-sell boundaries. The transaction cost function is assumed to be $C^2$ on $(0,\infty)$, nondecreasing, and convex on $(0,\infty)$, as formalized in Assumption~\ref{ass:cost}. Our main contributions are:
		\begin{itemize}
		\item We derive asymptotic formulas for the optimal trading boundaries and the optimal consumption strategy. 
			
		\item We recover fixed, proportional, fixed-plus-proportional, and nonlinear cost structures as special cases of the general framework.
		
		\item We show in Section~\ref{sec: NonLinear costs} that a no-trade region exists if and only if the cost function contains a fixed or proportional leading component; otherwise, the optimal strategy reduces to continuous trading along the Merton line.

		\item We establish a rigorous verification theorem for the exact impulse-control problem with a strictly positive fixed cost. Furthermore, we analyze its limiting transitions to singular and absolutely continuous control regimes as the fixed cost component vanishes.
	\end{itemize}

{\bf Notation:}  
Throughout this paper, we assume a probability space $(\Omega, \mathcal{F},\mathbb{P})$, where $\mathbb{P}$ is a probability measure on $\Omega$ and $\mathcal{F}$ is a $\sigma$-algebra. We endow the probability space $(\Omega, \mathcal{F},\mathbb{P})$ with an increasing filtration $\{\mathcal{F}_t\}_{t\geq0}$, which is a right-continuous filtration on $(\Omega, \mathcal{F})$ such that $\{\mathcal{F}_0\}$ contains all the $\mathbb{P}$-negligible subsets. For any integrable random variable or process $X$, $\mathbb{E}[X]$ denotes expectation with respect to $\mathbb{P}$. Uncertainty in the model is generated by a $\mathcal{F}_t$-adapted standard one-dimensional Brownian motion $W$. All stochastic integrals are defined in the sense of It\^{o}. 

Let $\mathcal{O} \subset \mathbb{R}^n$ be an arbitrary open domain. We denote by $C^2(\mathcal{O})$ the space of twice continuously differentiable functions on $\mathcal{O}$, and by $W^{2,\infty}_{loc}(\mathcal{O})$ the local Sobolev space of functions whose weak derivatives up to the second order are locally essentially bounded. A $C^{1,2}$-function denotes a function that is once continuously differentiable in the time variable and twice continuously differentiable in all spatial variables. ${\bf 1}_{\{\cdot\}}$ denotes the characteristic function. 

For a function $f(x_1,\dots,x_n)$, $\partial_{i}f$ means $\frac{\partial}{\partial x_i}f$ and $\partial_{ij}f=\frac{\partial^2}{\partial x_i\partial x_j}f$ with $i,j=1,2,\dots, n$. In some cases, we use $\partial_{\bullet}=\frac{\partial}{\partial \bullet}$ to emphasize that the partial derivative is taken with respect to the variable $\bullet$. We denote the first and second derivatives of the transaction cost function by $k'$ and $k''$, respectively.

The rest of this paper is organized as follows. In Section~\ref{sec: finaincal market}, we introduce the market model, present a general transaction cost structure, and review the optimal investment and consumption policies in the absence of transaction costs. In Section~\ref{sec: general costs}, we address the problem involving general transaction costs and present leading-order solutions for optimal asset-allocation strategy, consumption strategy, and the value function. Section~\ref{sec: Linear costs} presents the optimal investment and consumption strategies for linear costs, which represent a combination of fixed and proportional costs. Section~\ref{sec: NonLinear costs} extends the analysis to nonlinear transaction costs and shows that the existence of a no-trade region is governed by the fixed and proportional components of the cost function. Section~\ref{sec: verification} establishes a verification theorem for the exact impulse-control problem under a strictly positive fixed cost, and clarifies its limiting transitions to degenerate control regimes as the fixed cost vanishes. For the sake of readability, the details of the mathematical proofs are provided in Appendices \ref{appendix: proof general}-\ref{appendix: proof effect of costs}.

\section{The Financial Market, Transaction Costs, and the Investor's Objective}\label{sec: finaincal market}
\subsection{The Financial Market}
We consider a financial market consisting of a risk-free asset $B$ and a risky asset $S$ governed by
\begin{equation}\label{Bt St}
{\rm d}B(t)=rB(t){\rm d}t\ \ \text{and}\ \ {\rm d}S(t)=S(t)(\mu{\rm d}t+\sigma{\rm d}W(t)),
\end{equation}
where $r$ is the constant risk-free rate, $\mu$ (assuming $\mu>r$) is the constant expected return, and $\sigma$ is the volatility of the risky asset $S$.  Following \cite{Buraschi-Porchia-Trojani-JF-2010, French-Schwert-Stambaugh-JFE-1987, Kraft-QF-2005, Merton-JFE-1980}, we assume that the risk premium $A = \frac{\mu-r}{\sigma^2}$ is constant, which implies that the stock excess return is proportional to the stock variance.

\subsection{Transaction Costs and Wealth Dynamics}\label{sec:costs-and-wealth}

\noindent\textbf{Transaction Cost Function.} 
In this market, trading the risky asset incurs a transaction cost. When the investor trades $|\Delta n|$ units of the risky asset at the price $S$, the dollar amount traded is $\Delta = |\Delta n| S$, and the associated cost is $k(\Delta)$.
This convention is consistent with the standard proportional transaction-cost formulation. If an investor trades $q$ shares at price $S$, the dollar amount traded is $\Delta=|q|S$, and a proportional transaction cost is therefore of the form $k(\Delta)=k_2\Delta$. This corresponds to the usual bid-ask spread model, in which the transaction cost is proportional to the dollar amount traded. The present formulation is more general because $k(\cdot)$ may also include a fixed component. In that case, a fixed transaction cost is a dollar amount paid per trade and should not be represented as a price-dependent per-share quantity such as $k_1/S$. Thus, we model transaction costs directly as a function of the dollar amount traded.
Throughout the paper, we use $\Delta$ for the dollar size of a trade, while $\triangle$ is used for changes in portfolio proportions and the rescaled perturbation variables, such as $\triangle\phi$, $\triangle W_B$, and $\triangle Y_B$.
We impose the following assumptions on $k(\cdot)$:
\begin{Assumption}\label{ass:cost}
	The transaction cost function $k: [0, \infty) \to [0, \infty)$ is of class $C^2$ on $(0, \infty)$ and satisfies:
	\begin{enumerate}
		\item \textbf{Non-negativity}: $k(z) \ge 0$ for all $z \ge 0$. 
		We adopt the convention that $k(0)=0$ (no cost is paid when no trade occurs), while the right limit $\lim_{z\downarrow 0}k(z)$ may be strictly positive, allowing for a fixed cost component;
		\item\label{Monotonicity} \textbf{Monotonicity}: $k'(z) \ge 0$ for all $z > 0$;
		\item\label{Convexity} \textbf{(Optional) Convexity}: $k''(z) \ge 0$ for all $z > 0$.
	\end{enumerate}
\end{Assumption}

The following remarks clarify Assumption~\ref{ass:cost}.

\begin{Remark}\label{rem:cost-approx}
	\leavevmode\par
	\begin{enumerate}[label=(\alph*)]
		\item The convexity assumption in the present CRRA setting differs from the local convexity condition often used in models with exponential utility. Under exponential utility, the wealth variable can be separated from the value function, so the trading decision is governed directly by the dollar amount traded. In the small-transaction-cost regime, this dollar amount is itself small, and therefore only the local behavior of the transaction cost function near the origin is relevant. By contrast, under CRRA utility the problem is naturally expressed in terms of portfolio proportions. The trading decision depends on the ratio of the dollar amount traded to total wealth. This ratio may be small even when the dollar amount traded is not small in absolute terms, especially for large wealth. Hence the transaction cost function may be evaluated away from the origin. For this reason, we conditionally consider the convexity of \(k(\cdot)\) on \((0,\infty)\), rather than only local convexity near the origin.

		\item Condition~\ref{Monotonicity} rules out anomalous cost functions for which larger trades incur a strictly lower absolute total cost than smaller trades; such specifications would create pathological arbitrage incentives for artificially large transactions and would be inconsistent with the standard no-trade-region structure.
		
		\item\label{control regimes} {Condition~\ref{Convexity} is marked as optional because it dictates the mathematical nature of the optimal control formulation. While general convexity allows for proportional costs, strict convexity ($k''>0$) imposes strictly increasing marginal costs, which inherently penalizes massive block trades. However, strict convexity without a fixed cost implies strict superadditivity, which creates artificial incentives for infinite order splitting to minimize market impact, thereby shifting the problem into an absolutely continuous rate-control regime. By contrast, a strictly positive right limit $\lim_{z\downarrow 0}k(z) > 0$ introduces local subadditivity near the origin. This fixed cost component effectively prevents infinite order splitting from creating artificial cost advantages, guaranteeing that the optimal strategy consists of discrete block trades. This subadditivity is the essential prerequisite for the impulse-control verification theorem (Theorem \ref{verification-thm}) provided in Section \ref{sec: verification}.}

		\item The transaction cost function $k(\cdot)$ is of class $C^2$ only on $(0, \infty)$, which allows a fixed cost component. Our analysis requires \(k(\cdot)\) to be \(C^2\) on \((0,\infty)\). Therefore, transaction cost functions with kink points at positive trade sizes are outside the direct scope of the model, although a jump at the origin is allowed in order to accommodate fixed costs.

		\item Empirical evidence on market impact (e.g., \cite{Almgren-Thum-Risk-2005, Lillo-Farmer-Mantegna-Nature-2003}) indicates that the price impact of individual trades is typically a smooth, \textit{increasing}, and concave function of trade size. In our framework, the convexity of the total transaction cost function is perfectly consistent with this empirical evidence: since the \textit{variable component} of the total cost can be mathematically interpreted as the definite integral of such a positive and increasing price-impact function, its second derivative is positive, yielding a strictly convex cost profile.
	\end{enumerate}
\end{Remark}

\noindent\textbf{Wealth Dynamics.} 
Let $w^B(t)$ and $w^S(t)$ denote the dollar amounts invested in the risk-free and the risky asset, respectively. The investor's trading strategy is described by two nondecreasing, right-continuous, adapted cumulative processes $\mathfrak{L}(t)$ and $\mathfrak{J}(t)$ with $\mathfrak{L}(0) = \mathfrak{J}(0) = 0$, representing the total dollar amounts of purchases and sales of the risky asset up to time $t$. The consumption rate is denoted by $c(t) \ge 0$. 

Between trades, the wealth processes evolve according to the standard \textit{self-financing} dynamics. At a trading time $\tau$, a purchase of dollar amount $\Delta \mathfrak{L} > 0$ or a sale of dollar amount $\Delta \mathfrak{J} > 0$ results in an instantaneous change in the portfolio positions, with the transaction cost $k(\Delta \mathfrak{L} + \Delta \mathfrak{J})$ deducted from the risk-free account. Consequently, the wealth processes satisfy the following integral equations:
\begin{align}
	w^B(t) &= x_1 + \int_{0}^{t} \big(r w^B(s) - c(s)\big) \, ds - \mathfrak{L}(t) + \mathfrak{J}(t) - \sum_{\tau_i \le t} k(\Delta \mathfrak{L}_i + \Delta \mathfrak{J}_i), \label{def-wB} \\ 
	w^S(t) &= x_2 + \int_{0}^{t} \mu w^S(s) \, ds + \int_{0}^{t} \sigma w^S(s) \, dW(s) + \mathfrak{L}(t) - \mathfrak{J}(t), \label{def-wS}
\end{align}
where $\{\tau_i\}$ are the trading times of $\mathfrak{L}$ and $\mathfrak{J}$, and $\Delta \mathfrak{L}_i = \mathfrak{L}(\tau_i) - \mathfrak{L}(\tau_i-)$, $\Delta \mathfrak{J}_i = \mathfrak{J}(\tau_i) - \mathfrak{J}(\tau_i-)$ are the corresponding trade sizes.

\subsection{Solvency Region}

Define the \emph{solvency region} 
\begin{equation}\label{def:solvency}
	\mathcal{S} \triangleq \big\{ (x, y) \in \mathbb{R} \times [0, \infty) : x + y - k(y) > 0 \big\}.
\end{equation}
The solvency region consists of all portfolio positions such that if the investor were forced to liquidate the entire risky position immediately, the proceeds after paying the transaction cost $k(y)$ would be strictly positive. Since short-selling is prohibited, we impose $y \ge 0$. This set is {relatively open in $\mathbb{R} \times [0, \infty)$}, and its closure is denoted by $\bar{\mathcal{S}}$.

\subsection{Admissible Strategies}\label{sec:admissible}

When the transaction cost has a positive fixed component, trading strategies are naturally represented in impulse-control form, since frequent infinitesimal trading would be prohibitively costly. The pure proportional case, in which trading may occur at the boundary in a singular-control form, is recovered later as a limiting regime.

\begin{Definition}[Trading Strategy]\label{def:trading}
	A \emph{trading strategy} is a sequence of triples $\{(\tau_i, \Delta \mathfrak{L}_i, \Delta \mathfrak{J}_i)\}_{i \ge 1}$ such that:
	\begin{enumerate}
		\item $\{\tau_i\}_{i \ge 1}$ is a strictly increasing sequence of $\{\mathcal{F}_t\}$-stopping times with $\tau_i \to \infty$ almost surely;
		\item for each $i$, $\Delta \mathfrak{L}_i$ and $\Delta \mathfrak{J}_i$ are nonnegative, $\mathcal{F}_{\tau_i}$-measurable random variables representing the dollar amounts of purchases and sales at time $\tau_i$, respectively, and they satisfy $\Delta \mathfrak{L}_i \cdot \Delta \mathfrak{J}_i = 0$ almost surely.
	\end{enumerate}
	The corresponding cumulative purchase and sale processes are defined as
	\[
	\mathfrak{L}(t) = \sum_{\tau_i \le t} \Delta \mathfrak{L}_i, \qquad \mathfrak{J}(t) = \sum_{\tau_i \le t} \Delta \mathfrak{J}_i, \qquad t \in [0, T],
	\]
	with $\mathfrak{L}(0)=\mathfrak{J}(0)=0$.
\end{Definition}

\begin{Definition}[Admissible Strategy]\label{def:admissible}
	Let $(t, x, y) \in [0, T] \times \overline{\mathcal{S}}$ be an initial state. 
	\begin{enumerate}
		\item A \emph{strategy} is a triple $\pi = (c, \mathfrak{L}, \mathfrak{J})$ consisting of:
		\begin{enumerate}
			\item a nonnegative, progressively measurable consumption rate process $c = \{c(s)\}_{s \in [t, T]}$ satisfying $\int_t^T c(s) \, ds < \infty$ almost surely;
			\item a trading strategy $(\mathfrak{L}, \mathfrak{J})$ in the sense of Definition~\ref{def:trading}.
		\end{enumerate}
		Given such a strategy, the wealth processes $(w^B, w^S)$ evolve according to \eqref{def-wB}--\eqref{def-wS} on $[t, T]$.
		\item The strategy $\pi$ is called \emph{admissible} if the following conditions hold:
		\begin{enumerate}
			\item \textbf{Solvency}: $(w^B(s), w^S(s)) \in \overline{\mathcal{S}}$ for all $s \in [t, T]$ almost surely.
			\item\label{Integrability} \textbf{Integrability}: The consumption and terminal wealth satisfy
			\[
			\mathbb{E}\left[ \int_t^T e^{-\beta s} |u_1(c(s))| \, ds \right] < \infty, \qquad
			\mathbb{E}\left[ e^{-\beta T} |u_2(w^B(T) + w^S(T))| \right] < \infty.
			\]
			\item\label{Square-integrability} \textbf{Square-integrability}: 
			\[
			\mathbb{E}\left[ \int_t^T \big( w^S(s) \big)^2 \, ds \right] < \infty.
			\]
		\end{enumerate}
	\end{enumerate}
	The set of all such admissible strategies is denoted by $\mathcal{A}(t, x, y)$. We note that $\mathcal{A}(t, x, y) \neq \emptyset$ if and only if $(x, y) \in \mathcal{S}$ (see, e.g., \cite{Shreve-Soner-AAP-1994} for a proof in a related setting).
\end{Definition}
\begin{Remark}
	\label{rem:admissible-justification}
	Condition~(\ref{Square-integrability}) is imposed to ensure that the stochastic integrals appearing in the wealth dynamics and in the subsequent perturbation analysis are true martingales. In particular, it justifies the use of It\^o's formula for terms such as
	\[
	\int_t^{s} \sigma w^S(u)\,dW(u)
	\quad\text{and}\quad
	\int_t^{s} \sigma v_z(u,w(u))w^S(u)\,dW(u),
	\]
	where $v(t,z)$ denotes the frictionless value function in \eqref{mertoncrra V}. 
\end{Remark}

\begin{Remark}
	Even in the presence of super-linear costs (e.g., quadratic costs), the solvency constraint together with the integrability requirements automatically prevents the investor from executing arbitrarily large trades that would incur disproportionately high expenses. Consequently, the definition of admissible strategies remains well-posed and the asymptotic analysis developed in this paper is unaffected by the global growth behavior of $k(\cdot)$.
\end{Remark}

\subsection{The Investor's Objective}

An investor with initial wealth $w_0$ has a CRRA utility function given by 
\begin{equation}\label{crra}
	u(x)=
	\left\{\begin{aligned}
		&\frac{a x^{\gamma}}{\gamma}, && \gamma < 1,\ \gamma \neq 0,\\
		&a \log(x), && \gamma = 0.
	\end{aligned}\right.
\end{equation}
The following derivation is carried out under the assumption $\gamma\neq 0$. The logarithmic utility case $\gamma= 0$ can be obtained as the limiting case $\gamma \to 0$ of the results presented here. 
The investor allocates a fraction $\phi(t)$ of wealth to the risky asset and consumes at rate $c(t)$. Because trading is costly, continuous portfolio rebalancing is generally suboptimal. Instead, optimal trading is characterized by a no-trade region, described below and illustrated in Figure~\ref{fig:no-trade}. The degenerate case in which the no-trade region disappears is discussed in Section~\ref{sec: NonLinear costs}.

The investor's objective is to maximize the expected discounted utility from intermediate consumption and terminal wealth over all admissible strategies:
\begin{align}\label{objective}
	V(t, x, y) = \sup_{(c, \mathfrak{L}, \mathfrak{J}) \in \mathcal{A}(t, x, y)} \mathbb{E}\left[ \int_t^T \alpha e^{-\beta s} u_1(c(s)) \, ds + (1-\alpha) e^{-\beta T} u_2(w^B(T) + w^S(T)) \right],
\end{align}
The parameter $\alpha \in [0,1]$ controls the relative weight placed on intermediate consumption versus terminal wealth, while $a_c$ and $a_w$ normalize the two CRRA utility functions $u_1$ and $u_2$, respectively. The parameter $\beta>0$ is the subjective discount rate. The wealth processes $w^B$ and $w^S$ are driven by the chosen strategy according to \eqref{def-wB}--\eqref{def-wS}.

In the absence of transaction costs ($k(\cdot) \equiv 0$), the optimal portfolio strategy is to hold a constant proportion of total wealth in risky assets $\phi^* = \frac{\mu-r}{(1-\gamma)\sigma^2}$ (the Merton proportion), which for typical parameter values lies well inside the interval $(0,1)$, and continuous trading is optimal. For the general cost structure considered here, closed-form solutions are unavailable. We therefore develop an asymptotic expansion for {a small cost-to-wealth ratio}.

The optimal strategy is characterized by four free boundaries: the buy boundary $\phi_B$, the post-buy boundary $\widehat{\phi}_B$, the sell boundary $\phi_S$, and the post-sell boundary $\widehat{\phi}_S$. In general, these boundaries depend on the remaining horizon and current wealth. A schematic illustration is provided in Figure~\ref{fig:no-trade}. 
\begin{figure}[H]
	\centering
	\includegraphics[height=8.5cm, width=8.7cm]{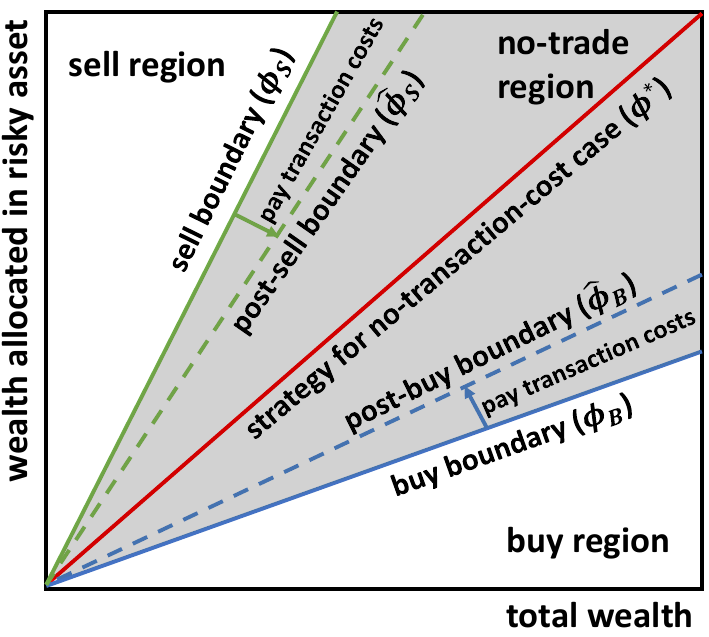}
	\caption{\textbf{A sketch of optimal asset allocation strategies with transaction costs.} The gray area represents the no-trade region, within which the frictionless optimal allocation $\phi^*$ (red line) is located. The four free boundaries are indicated: $\phi_B$ (buy), $\widehat{\phi}_B$ (post-buy), $\phi_S$ (sell), and $\widehat{\phi}_S$ (post-sell).}
	\label{fig:no-trade}
\end{figure}
Figure~\ref{fig:no-trade} can be interpreted as follows. Before making a trading decision, the investor weighs two competing factors: (1)~the benefit of rebalancing the portfolio toward the frictionless optimum, and (2)~the transaction costs incurred by trading. When the benefit from rebalancing is smaller than the associated costs, the investor refrains from trading; the portfolio is then said to lie within the no-trade region. Conversely, when the risky allocation falls below the buy boundary, the benefit of increasing the risky exposure outweighs the transaction costs, prompting the investor to buy. The boundary separating the no-trade region from the buy region is denoted by $\phi_B$. A purchase moves the portfolio allocation to the post-buy boundary $\widehat{\phi}_B$, which remains inside the no-trade region. The sell boundary $\phi_S$ and post-sell boundary $\widehat{\phi}_S$ admit a symmetric interpretation.

For the subsequent analysis, it is convenient to describe the portfolio in terms of the total wealth $w(t) = w^B(t) + w^S(t)$ and the fraction $\phi(t) = w^S(t) / w(t)$ invested in the risky asset. In terms of these variables, the value function is defined as
\begin{align}\label{crra value function definition}
	\widetilde V(t,w,\phi) = \sup_{(c, \mathfrak{L}, \mathfrak{J}) \in \mathcal{A}(t, (1-\phi)w, \phi w)} \mathbb{E}\left[ \int_t^T \alpha e^{-\beta s} u_1(c(s)) \, ds + (1-\alpha) e^{-\beta T} u_2(w(T)) \,\Big|\, \mathcal{F}_t \right],
\end{align}
where the supremum is taken over all admissible strategies in the sense of Definition~\ref{def:admissible}, with initial wealth components $x = (1-\phi)w$ and $y = \phi w$. The terminal condition is
\[
\widetilde V(T,w,\phi) = (1-\alpha)e^{-\beta T}u_2(w).
\]

Let $\tau = T - t$ denote the remaining investment horizon. Exploiting the homogeneity of the CRRA utility functions, we factor the value function as
\begin{equation}\label{crraV}
	\widetilde V(t,w,\phi) = V(\tau,w,\phi) = e^{-\beta (T-\tau)} \frac{a_w w^{\gamma}}{\gamma} \big[ f(\tau,w,\phi) \big]^{1-\gamma},
\end{equation}
where the function $f(\tau,w,\phi)$ is to be determined. The initial condition at $\tau = 0$ is 
\[
f(0,w,\phi) = (1-\alpha)^{\frac{1}{1-\gamma}},
\]
independent of both $w$ and $\phi$.
\begin{Remark}
	In the absence of transaction costs, the optimization problem is scale-invariant: multiplying initial wealth by a constant factor leaves the optimal portfolio weight and consumption-to-wealth ratio unchanged, while the value function scales by the corresponding CRRA power. Hence the frictionless factor $f^*$ depends only on the remaining horizon $\tau$. 
\end{Remark}

For future reference, we introduce several constants that will appear frequently in the asymptotic expansions:
\begin{align}\label{def_a}
	a_1 = \left(\frac{\alpha a_c}{a_w}\right)^{\frac{1}{1-\gamma}}, \quad
	a_2 = \frac{\gamma\sigma^2}{2} a_4^2 + \frac{\gamma r - \beta}{1-\gamma}, \quad
	a_3 = (1-\alpha)^{\frac{1}{1-\gamma}}, \quad
	a_4 = \frac{A}{1-\gamma}, \quad
	A = \frac{\mu-r}{\sigma^2}.
\end{align} 

\subsection{Optimal policies in the absence of transaction costs}\label{sec: no costs}
If the transaction-cost-to-wealth ratio is of order one, the trade is unlikely to be beneficial to the investor. Therefore, in practice, trades usually occur when this ratio is small. We therefore develop a singular perturbation expansion around the frictionless solution. The perturbation parameter is the transaction-cost-to-wealth ratio. For the reader's convenience, we briefly recall the classical Merton solution for the case $k(\cdot) \equiv 0$.

In the absence of transaction costs, the optimal investment and consumption strategies are given by
\begin{align}\label{mertoncrra phi c}
	\phi^* = \frac{A}{1-\gamma}, \qquad c^*(\tau,w) = \frac{a_1 w}{f^*(\tau)},
\end{align}
and the corresponding value function is
\begin{align}\label{mertoncrra V}
	V^*(\tau,w) = e^{-\beta (T-\tau)} \frac{a_w w^{\gamma}}{\gamma} \big[ f^*(\tau) \big]^{1-\gamma},
\end{align}
where
\begin{align}\label{mertoncrra f}
	f^*(\tau) = \left( \frac{a_1}{a_2} + a_3 \right) e^{a_2 \tau} - \frac{a_1}{a_2}.
\end{align}
The constants $a_1, a_2, a_3$ and $A$ are defined in \eqref{def_a}. Detailed derivations can be found in \cite{Merton-RES-1969, Merton-JET-1971, Merton-JFE-1980}.

The frictionless value function \eqref{mertoncrra V} and optimal policies \eqref{mertoncrra phi c} serve as the base point for our singular perturbation expansion for {a small cost-to-wealth ratio}.

\section{Optimal policies under a general transaction cost structure}\label{sec: general costs}
The frictionless policies in \eqref{mertoncrra phi c} require continuous portfolio rebalancing, which is no longer optimal when trading is costly. In this section, we analyze the transaction cost function $k(\cdot)$ introduced in Assumption~\ref{ass:cost}, where the dollar amount traded is $\Delta=\triangle\phi(t)w(t)$. The key assumptions are $C^2$ smoothness on $(0,\infty)$, monotonicity, and convexity on $(0,\infty)$. We also allow a fixed-cost component through the right-limit interpretation that the cost of an arbitrarily small positive trade may be strictly positive, while no cost is incurred when no trade is made.

As discussed in Section~\ref{sec: finaincal market} and illustrated in Figure~\ref{fig:no-trade}, transaction costs give rise to a no-trade region bounded by four free boundaries: the buy boundary $\phi_B$, the post-buy boundary $\widehat{\phi}_B$, the sell boundary $\phi_S$, and the post-sell boundary $\widehat{\phi}_S$. We now derive the governing equation inside the no-trade region and the boundary conditions that determine these four boundaries.

The degenerate case in which the no-trade region collapses to zero is discussed in Section~\ref{sec: NonLinear costs}; this case occurs when the transaction cost function contains only a nonlinear component.

\noindent{\bf Governing equation in the no-trade region:} 		
Inside the no-trade region, only the consumption policy, $c$, is a control variable. From \eqref{def-wB}-\eqref{def-wS} and $w(t) = w^B(t) + w^S(t)$, we have
\begin{align}
{\rm d}\phi(t)=\,&-{\rm d}(1-\phi(t))\notag\\
=\,&-{\rm d}\left(\frac{w^B(t)}{w(t)}\right)=-\frac{1}{w(t)}{\rm d}w^B(t)-w^B(t){\rm d}\left(\frac{1}{w(t)}\right)\notag\\
=\,&\left[\sigma^2\phi(t)(1-\phi(t))(A-\phi(t))+\frac{\phi(t)}{w(t)}c(t)\right]{\rm d}t+\sigma\phi(t)(1-\phi(t)){\rm d}W.\label{d phi}
\end{align}
Here we use the no-trading conditions, namely,
$\Delta\mathfrak{L}(t)=\Delta\mathfrak{J}(t) = \triangle\phi(t)\,w =0$.

The H-J-B equation for the value function, $V(\tau,w,\phi)$, is
\begin{align}\label{HJBVcrra}
\sup_{c}\Big\{&\alpha e^{-\beta (T-\tau)}u_1(c)-\partial_{1}V+[(r+A\sigma^2\phi)w-c]\partial_{2}V+\frac{1}{2}\sigma^2\phi^2w^2\partial_{22}V+\sigma^2\phi(1-\phi)(A-\phi)\partial_{3}V\notag\\
&+\frac{\phi}{w}c\,\partial_{3}V+\frac{1}{2}\sigma^2\phi^2(1-\phi)^2\partial_{33}V+\sigma^2\phi^2(1-\phi)w\partial_{23}V\Big\}=0,
\end{align}
with the initial condition $V(0,w,\phi)=(1-\alpha)e^{-\beta T}u_2(w)$.  
From \eqref{crra} and \eqref{crraV}, \eqref{HJBVcrra} can be expressed as
\begin{align}\label{HJBVcrra-f}
	\sup_{c}\Big\{&\frac{\alpha a_c}{a_w}f^{1-\gamma}\left(\frac{c}{w}\right)^{\gamma}-\big[\gamma+(1-\gamma)f^{-1}\big(w\partial_{2}f-\phi \partial_{3}f\big)\big]\left(\frac{c}{w}\right)+(1-\gamma)f^{-1}\Big[-\partial_{1}f+h_2(\tau,\phi)f\Big]\notag\\
	&+(1-\gamma)f^{-1}\Big[\big[h_1(\tau,\phi)-\sigma^2\phi^2\big](1-\phi)\partial_{3}f+\big[r+h_1(\tau,\phi)\big]w\,\partial_{2}f\Big]\notag\\
	&+(1-\gamma)f^{-1}\Big[\frac{1}{2}\sigma^2\phi^2(1-\phi)^2\Big(\partial_{33}f-\gamma \frac{(\partial_{3}f)^2}{f}\Big)+\frac{1}{2}\sigma^2\phi^2w^2\Big(\partial_{22}f-\gamma \frac{(\partial_{2}f)^2}{f}\Big)\Big]\notag\\
	&+(1-\gamma)f^{-1}\Big[\sigma^2\phi^2(1-\phi)w\Big(\partial_{23}f-\gamma\frac{\partial_{3}f\partial_{2}f}{f}\Big)\Big]
	\Big\}=0,
\end{align}
where 
\vspace*{-10pt}
\begin{align*}
	&h_1(\tau,\phi)=A\sigma^2\phi+\gamma\sigma^2\phi^2,\\ 
	&h_2(\tau,\phi)=\frac{\gamma r-\beta}{1-\gamma}+\frac{\gamma}{1-\gamma}\sigma^2A\phi-\frac{1}{2}\gamma\sigma^2\phi^2,
\end{align*} 
then the optimal consumption strategy, expressed in terms of $f=f(\tau,w,\phi)$, is
\begin{equation}\label{crraoptimalc}
c(\tau,w,\phi)=a_1wf^{-1}\left[1+\frac{1-\gamma}{\gamma}\frac{w\,\partial_{2}f-\phi\, \partial_{3}f}{f}\right]^{\frac{1}{\gamma-1}}.
\end{equation}
After substituting \eqref{crraoptimalc} into \eqref{HJBVcrra-f}, we obtain the following governing equation inside the no-trade region 
\vspace*{-10pt}
\begin{align}\label{crraH-J-BVc*}
&-\partial_{1}f+\left[h_1(\tau,\phi)-\sigma^2\phi^2\right](1-\phi)\partial_{3}f+\left[r+h_1(\tau,\phi)\right]w\,\partial_{2}f+\frac{1}{2}\sigma^2\phi^2(1-\phi)^2\left(\partial_{33}f-\gamma \frac{(\partial_{3}f)^2}{f}\right)\notag\\
&+\frac{1}{2}\sigma^2\phi^2w^2\left(\partial_{22}f-\gamma \frac{(\partial_{2}f)^2}{f}\right)+\sigma^2\phi^2(1-\phi)w\left(\partial_{23}f-\gamma\frac{\partial_{3}f\partial_{2}f}{f}\right)\notag\\
&+h_2(\tau,\phi)f+a_1h_3(\tau,w,\phi,f)=0,
\end{align}
with initial condition $f(0,w,\phi)=a_3$, where $a_1$, $a_3$ are defined in \eqref{def_a}, and
\begin{align*}
&h_3(\tau,w,\phi,f)=\left[1+\frac{1-\gamma}{\gamma}\frac{w\,\partial_{2}f-\phi\, \partial_{3}f}{f}\right]^{\frac{\gamma}{\gamma-1}}.
\end{align*} 

\noindent{\bf Boundary conditions for the four free boundaries.}
At a trading time, the value function is continuous across the trade, and optimality requires the marginal value of the trade to balance its marginal transaction cost. We first state these conditions on the buy side. 

When the investor buys at the boundary $\phi_B$, the risky-asset proportion is shifted to the post-buy boundary $\widehat{\phi}_B$. The corresponding post-trade wealth is
\[
\widehat{w}_b = w - k\big((\widehat{\phi}_B-\phi_B)w\big).
\]
Value matching at the buy boundary gives
\begin{align}
	V(\tau,w,\phi_B)=&V(\tau,\widehat{w}_b,\widehat{\phi}_B), \label{buy_value}
\end{align}
which, via the factorization \eqref{crraV}, becomes
\begin{align}
	f(\tau,w,\phi_B)=&\left(\frac{\widehat{w}_b}{w}\right)^{\frac{\gamma}{1-\gamma}}f(\tau,\widehat{w}_b,\widehat{\phi}_B). \label{crrabc1}
\end{align}

We optimize the expected utility by choosing when to buy ($\phi_B$) and how much to buy ($\widehat{\phi}_B-\phi_B$). Applying the variation principle to \eqref{crrabc1} with respect to $\phi_B$ and $\widehat{\phi}_B$, respectively, we obtain
\begin{align}
	\partial_{\phi_B}f(\tau,w,\phi_B)
	=\,&k'\big((\widehat{\phi}_B-\phi_B)w\big)w\left\{\frac{\gamma}{1-\gamma}\frac{1}{\widehat{w}_b}f(\tau,w,\phi_B)+\left(\frac{\widehat{w}_b}{w}\right)^{\frac{\gamma}{1-\gamma}}\partial_{2}f(\tau,\widehat{w}_b,\widehat{\phi}_B)\right\},\label{crrabc2}
	\\
	\partial_{\widehat{\phi}_B}f(\tau,\widehat{w}_b,\widehat{\phi}_B)=\,&k'\big((\widehat{\phi}_B-\phi_B)w\big)w\left\{\partial_{2}f(\tau,\widehat{w}_b,\widehat{\phi}_B)+\frac{\gamma}{1-\gamma}\left(\frac{\widehat{w}_b}{w}\right)^{-\frac{1}{1-\gamma}}\frac1wf(\tau,w,\phi_B)\right\}.\label{crrabc3}
\end{align}
The three boundary conditions, \eqref{crrabc1}--\eqref{crrabc3}, determine $\phi_B$ and $\widehat{\phi}_B$.

Similarly, the following three equations determine $\phi_S$ and $\widehat{\phi}_S$:
\begin{align}
	f(\tau,w,\phi_S)=\, &\left(\frac{\widehat{w}_s}{w}\right)^{\frac{\gamma}{1-\gamma}}f(\tau,\widehat{w}_s,\widehat{\phi}_S),\label{crrasc1}
	\\
	\partial_{\phi_S}f(\tau,w,\phi_S)=\, &-k'\big((\phi_S-\widehat{\phi}_S)w\big)w\left\{\frac{\gamma}{1-\gamma}\frac{1}{\widehat{w}_s}f(\tau,w,\phi_S)+\left(\frac{\widehat{w}_s}{w} \right)^{\frac{\gamma}{1-\gamma}}\partial_{2}f(\tau,\widehat{w}_s,\widehat{\phi}_S)\right\},\label{crrasc2}
	\\
	\partial_{\widehat{\phi}_S}f(\tau,\widehat{w}_s,\widehat{\phi}_S)=\, &-k'\big((\phi_S-\widehat{\phi}_S)w\big)w\left\{\partial_{2}f(\tau,\widehat{w}_s,\widehat{\phi}_S)+\frac{\gamma}{1-\gamma}\left(\frac{\widehat{w}_s}{w} \right)^{-\frac{1}{1-\gamma}}\frac1wf(\tau,w,\phi_S)\right\},\label{crrasc3}
\end{align} 
where $w$ is the wealth before selling and $\widehat{w}_s =w+\triangle w_s$ is the wealth after selling. The sale changes wealth by $\triangle w_s=-k\big((\phi_S-\widehat{\phi}_S)w\big)$ due to the transaction cost associated with the sale.

In summary, to obtain optimal investment and consumption strategies, we solve the nonlinear partial differential equation \eqref{crraH-J-BVc*} in the no-trade region together with six boundary conditions---three associated with the buy and post-buy positions \eqref{crrabc1}--\eqref{crrabc3}, and three associated with the sell and post-sell positions \eqref{crrasc1}--\eqref{crrasc3}.

The boundary formulas below show that the post-trading boundaries lie inside the no-trade region, satisfying $\phi_B \leq \widehat{\phi}_B \leq \phi^* \leq \widehat{\phi}_S \leq \phi_S$, where $\phi^*$ is the frictionless Merton proportion \eqref{mertoncrra phi c}. Hence the no-trade region is $\mathcal{NT} = \{ \phi \in [0,1] : \phi_B(\tau,w) \leq \phi \leq \phi_S(\tau,w) \}$.

The expansion applies when the normalized cost-to-wealth ratio is small, i.e., $\widehat{k}(\triangle\phi,w) \triangleq \frac{k(\triangle\phi\,w)}{w}\ll1$. We refer to $\widehat{k}(\cdot,\cdot)$ as the normalized transaction cost. We apply a singular perturbation expansion $\phi = \phi^* + O\big((\widehat{k}(\triangle\phi,\,w))^\lambda\big)$ to determine the solution within the no-trade region and the free boundaries. The exponent $\lambda$ is determined by the local behavior of the cost function $k(\cdot)$ at the origin. For purely proportional costs ($\lim_{z\downarrow0}k(z)=0$, $k'(0+)>0$), one has $\lambda = 1/3$ \cite{Possamai-Soner-Touzi-CPDE-2015}. For costs with a fixed component ($\lim_{z\downarrow0}k(z)>0$), regardless of the presence of a proportional component, the leading-order scaling is governed by the constant term, yielding $\lambda = 1/4$ \cite{Altarovici-Muhle-Soner-FS-2015}. Since our general framework allows $\lim_{z\downarrow0}k(z) \ge 0$, the dominant asymptotic regime is determined by whether this right limit vanishes. The expansion naturally encompasses both regimes, and our leading-order solution unifies the pure fixed-cost, pure proportional-cost, and mixed cases. The derivations are presented in Appendix \ref{appendix: proof general}. The main results for the general cost structure are summarized in Theorem~\ref{General}; the special cases of purely fixed, purely proportional, and linear costs are recovered in Corollaries~\ref{CRRA fixed}, \ref{CRRA proportional}, and~\ref{CRRA linear}, respectively.

\begin{Theorem}[Leading-order solution under general transaction costs]\label{General}
	Let the transaction cost function $k(\cdot)$ satisfy Assumption \ref{ass:cost} and let the total wealth $w$ be such that the normalized cost-to-wealth ratio $\widehat{k}(\triangle\phi,w) \triangleq k(\triangle\phi\,w)/w \ll 1$ for the relevant trade-size fractions $\triangle\phi$. Then the leading-order solutions to \eqref{crraH-J-BVc*} subject to the boundary conditions \eqref{crrabc1}--\eqref{crrasc3} are the following:
	\begin{enumerate}
		\item The optimal buy boundary $\phi_B$, post-buy boundary $\widehat{\phi}_B$, sell boundary $\phi_S$, and post-sell boundary $\widehat{\phi}_S$ are given by
		\begin{align}
			&\phi_S=\phi^*+\tfrac{1}{2}\left(\widetilde{x}+x\right),\  \widehat{\phi}_S=\phi^*+\tfrac{1}{2}\left(\widetilde{x}-x\right),\label{optimal sell boundary}\\
			&\phi_B=\phi^*-\tfrac{1}{2}\left(\widetilde{x}+x\right),\  \widehat{\phi}_B=\phi^*-\tfrac{1}{2}\left(\widetilde{x}-x\right).\label{optimal buy boundary}
		\end{align}
			Here $x$ and $\widetilde{x}$ are determined by a selected nonnegative solution of the system of equations
		\begin{align}\label{original generalx}
			\left(\widetilde{x}^{2}-x^{2}\right)\widetilde{x}=A_2 (x,w) ,\ x^3\widetilde{x}=A_1(x,w),
		\end{align}
		where 
		\begin{equation}\label{general A1 A2}
			\left\{\begin{aligned}
				&A_1(x,w)=H\Big[\frac{k(xw)}{w}-xk'(xw)\Big],\  A_2(x,w)=Hk'(xw),\\ 
				& H=12\tfrac{(\phi^{*})^2}{1-\gamma}\left(1-\tfrac{A}{1-\gamma}\right)^2.
			\end{aligned}\right.
		\end{equation} 
		
			When $x\neq 0$, $x$ and $\widetilde{x}$ are given by
			\begin{align}
				&A_2(x,w) x^{9}+A_1 (x,w) x^{8}=A_1^3(x,w),\label{x}\\
				&\widetilde{x}=\frac{A_1(x,w)}{x^3}.\label{xtilde}
			\end{align}
			For a genuinely general cost function, \(A_1(x,w)\) and \(A_2(x,w)\) depend on the unknown \(x\), so \eqref{x} is an implicit scalar equation rather than a polynomial with fixed coefficients. In the special fixed, proportional, and fixed-plus-proportional cases, \(A_1\) and \(A_2\) reduce to constants and the relevant nonnegative solution is unique, as discussed below.
		
		\item The optimal value function is
		\begin{align}
			V(\tau,w)=V^*(\tau,w)\left[1+(1-\gamma)\frac{\widehat{f}(\tau,w)}{f^{*}(\tau)}\right].
			\label{crraleadingV}
		\end{align}
		\item The optimal consumption rate is
		\begin{align}
			c(\tau,w)=&c^*(\tau,w)-\frac{a_1w}{[f^*(\tau)]^{2}}\left[\widehat{f}(\tau,w)+\frac{1}{\gamma}w\,\partial_{2}\widehat{f}(\tau,w)\right].\label{crraleadingc}
		\end{align}
	\end{enumerate}
	In \eqref{crraleadingV} and \eqref{crraleadingc}, $\widehat{f}(\tau,w)$ is 
	\begin{align}
		&\widehat{f}(\tau,w)=-\gamma\sigma^2 \psi(\tau) \left(3\widetilde{x}^2+x^2\right),\label{crraV2}
	\end{align}
	where
	\begin{align}
		\psi(\tau)
		=\,\frac{1}{24}\left[e^{a_2\tau}\frac{a_2(a_1+a_2a_3)\tau-a_1}{a_2^2}+\frac{a_1}{a_2^2}\right].
		\label{hatI}
	\end{align}
	The constants $a_1, a_2$ and $a_3$ are defined in \eqref{def_a}. 
	
	In \eqref{optimal sell boundary}--\eqref{hatI}, $\phi^*$, $c^*$, $f^*$, and $V^*$ are the frictionless solutions, given by \eqref{mertoncrra phi c}, \eqref{mertoncrra f}, and \eqref{mertoncrra V}, respectively.	
\end{Theorem}
\begin{proof}
	See Appendix \ref{appendix: proof general}.  
\end{proof}

\begin{Remark}
Section~\ref{sec: verification} gives a classical verification theorem for the exact QVI formulation of the control problem. That theorem is conditional on the existence of a sufficiently smooth solution satisfying the stated solvency, growth, intervention, and admissibility conditions. A separate convergence theorem showing that the leading-order asymptotic formulas in Theorem~\ref{General} converge to the exact QVI solution is not studied here.
\end{Remark}

	The formulas below concern the nontrivial case in which the selected solution of \eqref{original generalx} is not identically zero. If \(A_1(x,w)=A_2(x,w)=0\) at a positive trade size \(x\), then \(k(xw)=0\) and \(k'(xw)=0\) at that trade size. The fully frictionless model is one such case. The degenerate case \(x=\widetilde{x}=0\), which can occur for purely nonlinear convex costs with no fixed or proportional leading component, is treated separately in Section~\ref{sec: NonLinear costs}.

\begin{Remark}\label{Remark special cases of original generalx}
		We now record the explicit solutions in the special cases where the effective quantities \(A_1\) and \(A_2\) can be treated as nonnegative constants, as happens for fixed and proportional leading components.
	\begin{enumerate}[label=(\alph*)]
			\item\label{A2 is 0} When $A_1 \neq 0$ and $A_2=0$ are constants, \eqref{original generalx} becomes $\left(\widetilde{x}^{2}-x^{2}\right)\widetilde{x}=0$, $x^3\widetilde{x}=A_1$. This gives 
		\begin{equation}\label{A1}
			x=\widetilde{x}=A_1^{\frac14}.
		\end{equation}

			\item\label{A1 is 0} When $A_1=0$ and $A_2 \neq0$ are constants, \eqref{original generalx} becomes $\left(\widetilde{x}^{2}-x^{2}\right)\widetilde{x}=A_2$, $x^3\widetilde{x}=0$. This gives 
		\begin{equation}\label{A2}
			x=0,\ \widetilde{x}=A_2^{\frac13}.
		\end{equation}
	\end{enumerate}
	\end{Remark}
	\begin{Remark}
		In cases \ref{A2 is 0} or \ref{A1 is 0}, the obtained solution is nonnegative and unique. 
		When \(A_1\) and \(A_2\) are nonnegative constants, as in the fixed-plus-proportional case, Descartes' rule of signs applied to \eqref{x} gives a unique positive real root for \(x\), and then \eqref{xtilde} determines a unique positive \(\widetilde{x}\). For a genuinely nonlinear cost function, however, \(A_1(x,w)\) and \(A_2(x,w)\) depend on \(x\); in that case \eqref{x} is an implicit equation, and uniqueness is not asserted without additional monotonicity conditions.
	\end{Remark}
\begin{Remark}
	Since {the normalized cost-to-wealth quantities are small}, namely, $A_1(x,w)\ll 1$, $A_2(x,w)\ll 1$, from \eqref{optimal sell boundary} and \eqref{optimal buy boundary}, 
	the post-sell boundary lies below the sell boundary, and the post-buy boundary lies above the buy boundary, namely $\phi_B\leq\widehat{\phi}_B\leq\phi^*\leq\widehat{\phi}_S\leq\phi_S$. In other words, the post-trading boundaries are inside the no-trade region, as illustrated in Figure \ref{fig:no-trade}.
\end{Remark}
\begin{Remark}
		The equations in \eqref{original generalx} show that a non-degenerate no-trade width requires a nonzero nonnegative contribution from the leading fixed or proportional component of the transaction cost. The purely nonlinear case, for which the leading-order no-trade width collapses to zero, is discussed separately in Section~\ref{sec: NonLinear costs}.
	\end{Remark}
Transaction costs play an important role in investment and consumption strategies. In particular, it is crucial to understand how transaction costs affect the buy boundary $\phi_B$, post-buy boundary $\widehat{\phi}_B$, sell boundary $\phi_S$, and post-sell boundary $\widehat{\phi}_S$. From \eqref{general A1 A2}, the normalized (dimensionless) cost $\frac{k(xw)}{w}$ can be expressed as
\begin{align*}
	\frac{k(xw)}{w} = \frac{1}{H}\big(A_1(x,w) + xA_2(x,w)\big).
\end{align*}
The term $\frac{xA_2(x,w)}{H}$ can be interpreted as an effective proportional cost with a trading-size-dependent coefficient $\frac{A_2(x,w)}{H}$, while $\frac{A_1(x,w)}{H}$ represents the portion of the transaction cost with the effective proportional cost removed. We now characterize how $A_1$ and $A_2$ affect $\phi_B$, $\widehat{\phi}_B$, $\phi_S$, and $\widehat{\phi}_S$.
\begin{Theorem}\label{effect of costs}
	Let $\phi_B$, $\widehat{\phi}_B$, $\phi_S$, and $\widehat{\phi}_S$ be, respectively, the optimal buy, post-buy, sell, and post-sell boundaries specified in Theorem \ref{General} for given $\tau$, $w$, $A_1(x,w)$ and $A_2(x,w)$. Then  
	\begin{enumerate}[label=(\alph*)]
			\item\label{effect of A1} For fixed $A_2(x,w)$, $\phi_S$ increases and $\phi_B$ decreases with $A_1(x,w)$. Therefore, the no-trade region widens as $A_1(x,w)$ increases. Furthermore, $\widehat{\phi}_S$ decreases and $\widehat{\phi}_B$ increases with $A_1(x,w)$.  Consequently, both the buy trade size $|\widehat{\phi}_B-\phi_B|$ and the sell trade size $|\widehat{\phi}_S-\phi_S|$ increase with $A_1(x,w)$.	
			\item\label{effect of A2} For fixed $A_1(x,w)$, $\phi_S$ increases and $\phi_B$ decreases with $A_2(x,w)$. Therefore, the no-trade region widens as $A_2(x,w)$ increases. Although $\widehat{\phi}_S$ increases and $\widehat{\phi}_B$ decreases with $A_2(x,w)$, both $|\widehat{\phi}_B-\phi_B|$ and  $|\widehat{\phi}_S-\phi_S|$ decrease with $A_2(x,w)$. Thus, both the buy and sell trade sizes decrease with $A_2(x,w)$.
	\end{enumerate}
\end{Theorem}
\begin{proof}
	See Appendix \ref{appendix: proof effect of costs}.
\end{proof}

\section{Optimal policies for fixed and proportional transaction costs}\label{sec: Linear costs}

Theorem~\ref{General} provides the leading-order optimal policies for any transaction cost function $k(\cdot)$ satisfying Assumption~\ref{ass:cost}. In this section, we specialize Theorem~\ref{General} to cost functions of the form
\begin{align}\label{linear costs}
	k(\triangle\phi\,w)=k_1+k_2(\triangle\phi\,w),
\end{align}
where $k_1 \ge 0$ and $k_2 \ge 0$, with at least one of them strictly positive. {The expression in \eqref{linear costs} is understood for positive trade size $\triangle\phi\,w>0$; when no trade occurs, no transaction cost is paid, so $k(0)=0$ by convention, while the right limit satisfies $\lim_{z\downarrow0}k(z)=k_1$ if $k_1>0$, as allowed in Assumption~\ref{ass:cost}.} This specification encompasses both fixed and proportional transaction cost components and is widely studied in the literature \cite{Morton-Pliska-MF-1995, Altarovici-Muhle-Soner-FS-2015, Janecek-Shreve-FS-2004, Liu-JF-2004}.

All specifications in \eqref{linear costs}---pure proportional ($k_1=0$, $k_2>0$), pure fixed-cost ($k_1>0$, $k_2=0$), and fixed-plus-proportional ($k_1>0$, $k_2>0$)---are directly admissible under Assumption~\ref{ass:cost}. The following corollaries are therefore valid in all these cases.

For the cost structure \eqref{linear costs}, we evaluate the quantities $A_1(x,w)$ and $A_2(x,w)$ appearing in \eqref{general A1 A2}. Noting that $k(\bullet)=k_1+k_2\bullet$ implies $k'(\bullet)=k_2$, we obtain
\[
A_1(w)=H\frac{k_1}{w}, \qquad A_2=Hk_2,
\]
where $H$ is defined in \eqref{general A1 A2}. Substituting these expressions into Theorem~\ref{General} yields the following corollary.

\begin{Corollary}[Fixed plus proportional transaction costs]\label{CRRA linear}
	A superscript ``$\mathrm{fp}$" denotes quantities associated with the fixed-plus-proportional cost structure \eqref{linear costs} when $k_1, k_2 \neq 0$.
	
	For the linear transaction costs given by \eqref{linear costs}, the optimal trading boundaries are 
	\begin{align}
		\phi^{\mathrm{fp}}_S=&\phi^*+\frac{1}{2}\left(z^3+z^{-1}\right)\big(H\frac{k_1}{w}\big)^{\frac14},\quad \widehat{\phi}^{\mathrm{fp}}_S=\phi^*+\frac{1}{2}\left(z^3-z^{-1}\right)\big(H\frac{k_1}{w}\big)^{\frac14},\label{crra optimal sell boundary linear}\\
		\phi^{\mathrm{fp}}_B=&\phi^*-\frac{1}{2}\left(z^3+z^{-1}\right)\big(H\frac{k_1}{w}\big)^{\frac14},\quad \widehat{\phi}^{\mathrm{fp}}_B=\phi^*-\frac{1}{2}\left(z^3-z^{-1}\right)\big(H\frac{k_1}{w}\big)^{\frac14},\label{crra optimal buy boundary linear}
	\end{align}
	where $z$ is the unique nonnegative real root of the following equation
	\begin{align}\label{original generalz}
		z^{9}-z=H^{\frac14}\big(\frac{k_1}{w}\big)^{-\frac34}k_2.
	\end{align}

	The optimal value function and consumption rate are
	\begin{align}
		V^{\mathrm{fp}}(\tau,w)=\,&V^*(\tau,w)\left[1+(1-\gamma)\frac{\widehat{f}^{\mathrm{fp}}(\tau,w)}{f^{*}(\tau)}\right],
		\label{crraleadingV linear}
				\\
	c^{\mathrm{fp}}(\tau,w)=\,&c^*(\tau,w)+ a_1{\sigma^2}\big[3\gamma z^6+(\gamma-2)z^{-2}\big]\frac{\psi(\tau)w}{[f^*(\tau)]^{2}}\big(H\frac{k_1}{w}\big)^{\frac12}.\label{c_rate linear}
\end{align}
	In \eqref{crraleadingV linear},
	\begin{align}
		&\widehat{f}^{\mathrm{fp}}(\tau,w)=-{\gamma\sigma^2} \psi(\tau) 
		\left(3z^6+z^{-2}\right)\big(H\frac{k_1}{w}\big)^{\frac12},\label{crraf2}
	\end{align}
	where $\psi(\tau)$ is given by \eqref{hatI}, and $a_1$ is defined in \eqref{def_a}. 
	
	In \eqref{crra optimal sell boundary linear}-\eqref{crraf2}, $V^*$, $f^*$, $\phi^*$, and $c^*$ denote the frictionless solutions given by \eqref{mertoncrra V}, \eqref{mertoncrra f}, and \eqref{mertoncrra phi c}. 
\end{Corollary}
\begin{proof} 
	This is a special case of Theorem \ref{General} with $A_1(w)=H\frac{k_1}{w},\ A_2=Hk_2$.
\end{proof}

The four free boundaries, $\phi_B$, $\widehat{\phi}_B$,  $\phi_S$, $\widehat{\phi}_S$, do not exhibit explicit time dependence. They only implicitly depend on time through $w$ in $A_1(w)=H\frac{k_1}{w}$, since $w$ changes with time.

To illustrate the effects of fixed and proportional costs, we provide a visual illustration for analytical expressions of optimal trading boundaries \eqref{crra optimal sell boundary linear}-\eqref{crra optimal buy boundary linear} from Corollary~\ref{CRRA linear}, showing how the four free boundaries depend on the fixed-cost-to-wealth ratio $k_1/w$ and the proportional cost rate $k_2$, see Figure \ref{fix prop}. 

In Figure \ref{fix prop}, as well as in Figures \ref{given-prop fix}-\ref{consumption fix prop}, the parameters are time discount rate $\beta=0.01$, risk-free rate $r = 0.01$, expected return $\mu = 0.069$, return volatility $\sigma = 0.22$, and risk-aversion parameter $\gamma = -1$. The investment horizon is set at $\tau=5$ years. These parameters are taken from \cite{Liu-JF-2004, Liu-Loewenstein-RFS-2002}.
\begin{figure}[htbp]
	\centering
\includegraphics[width=6in]{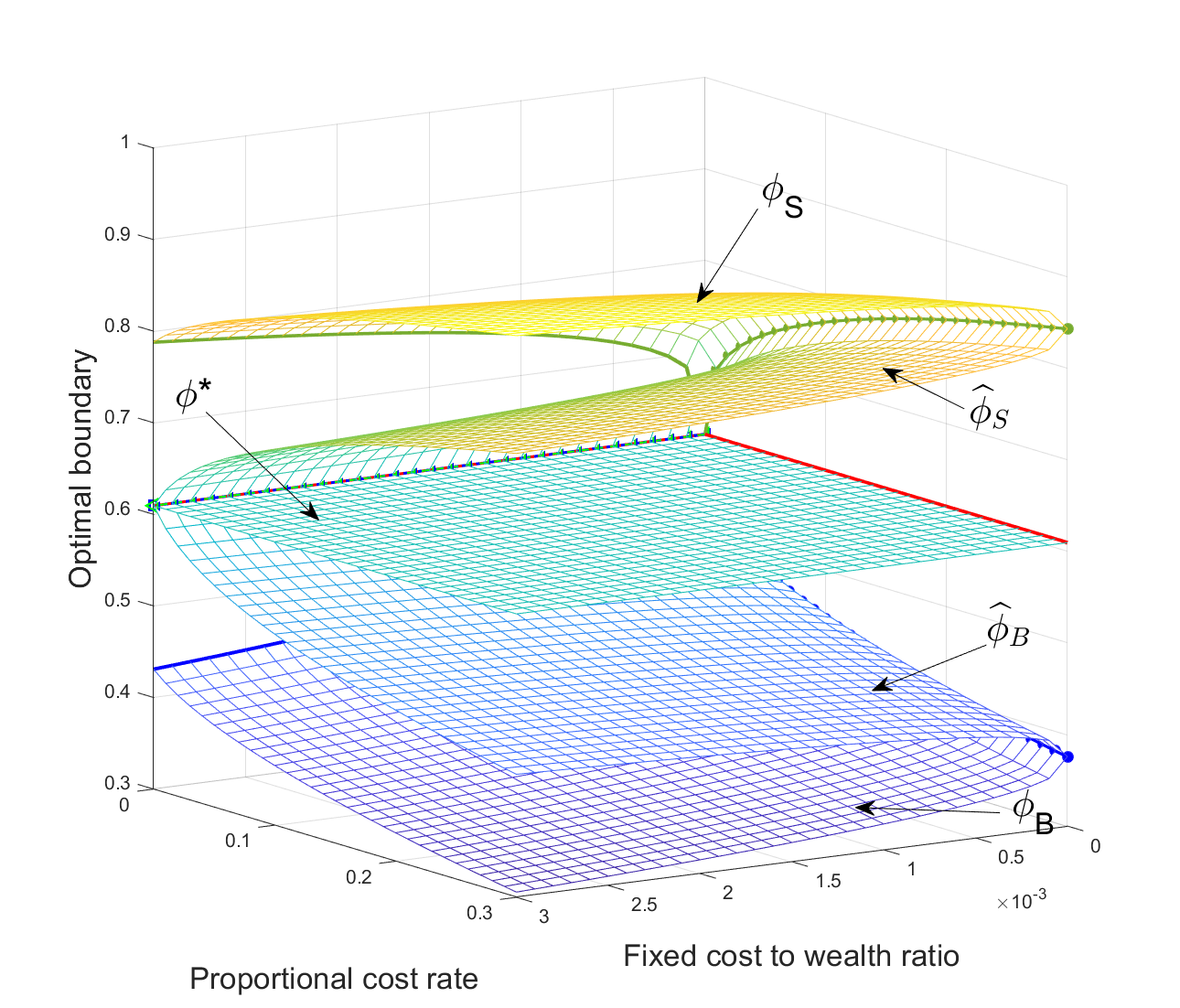}
	\caption{{\bf Four free boundaries as functions of fixed and proportional costs.} {\small The graph plots the optimal buy boundary $\phi_B$, post-buy boundary $\hat{\phi}_B$, sell boundary $\phi_S$, post-sell boundary $\hat{\phi}_S$ and Merton line $\phi^*$ against the fixed-cost-to-wealth ratio $\frac{k_1}{w}$ and the proportional cost rate $k_2$.}}
\label{fix prop}
\end{figure}

Figure \ref{given-prop fix} shows, for a given proportional cost rate $k_2$, how the free boundaries $\phi_B$, $\hat{\phi}_B$, $\phi_S$, and $\hat{\phi}_S$ vary with the fixed-cost-to-wealth ratio $\frac{k_1}{w}$. As $\frac{k_1}{w}$ increases, corresponding to either a higher fixed cost $k_1$ or lower wealth $w$, both the no-trade region and the trading size increase. A higher $\frac{k_1}{w}$ means that a larger fraction of wealth is spent on the fixed cost in each trade, so the investor trades less frequently. Consequently, the no-trade region becomes wider and the post-trade allocation lies deeper within the no-trade region after each trade. Figures \ref{p-a} and \ref{p-b} demonstrate that, for the same $\frac{k_1}{w}$, the trading size in Figure \ref{p-b} is smaller than that in Figure \ref{p-a} because $k_2$ is larger in Figure \ref{p-b}.
\begin{figure}[htbp]
	\centering
	\subfloat[]{\label{p-a}
		\includegraphics[width=3in,page={1}]{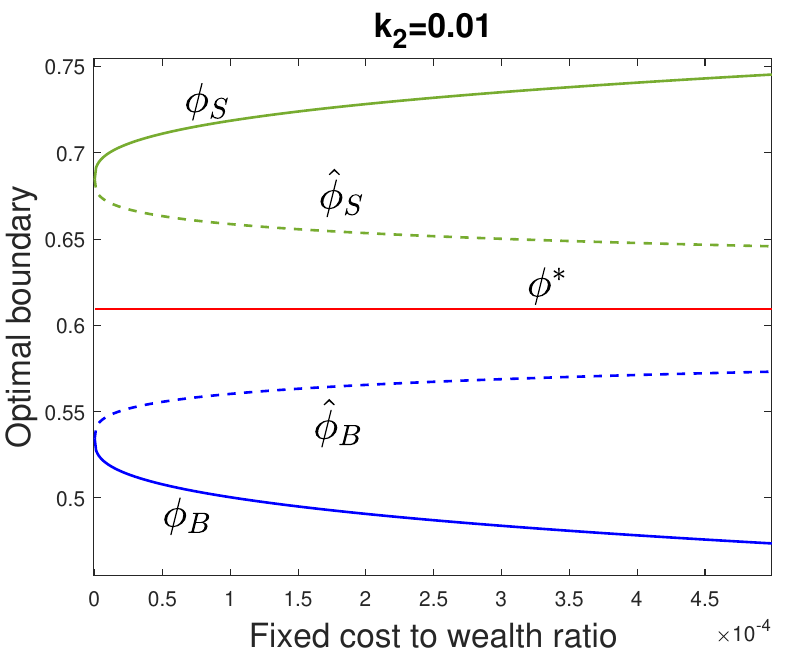}}
	\quad 
	\subfloat[]{\label{p-b} 		
		\includegraphics[width=3in,page={2}]{given_prop.pdf}}
	\caption{{\bf Given proportional cost rate $k_2$, four free boundaries as functions of the fixed-cost-to-wealth ratio $\frac{k_1}{w}$.} {\small The graph plots the optimal buy boundary $\phi_B$, post-buy boundary $\hat{\phi}_B$, sell boundary $\phi_S$, post-sell boundary $\hat{\phi}_S$ and Merton line $\phi^*$ against the fixed-cost-to-wealth ratio $\frac{k_1}{w}$ (i) corresponds to $k_2=0.01$, (ii) corresponds to $k_2=0.1$.}}
\label{given-prop fix}
\end{figure}

Figure \ref{given-fix prop} illustrates, for a fixed $\frac{k_1}{w}$, how the four free boundaries vary with the proportional cost rate $k_2$. Since $k_1$ is fixed while $w$ changes over time, for a given realization of the movement of the risky asset price, $w(t)$ is a random process, which implies that $\frac{k_1}{w}$ is also a random process. The projection of this random process onto the surfaces of $\phi_B$, $\hat{\phi}_B$, $\phi_S$, and $\hat{\phi}_S$ results in four random trajectories. However, Figure \ref{given-fix prop} provides insight into how the proportional cost affects $\phi_B$, $\hat{\phi}_B$, $\phi_S$, and $\hat{\phi}_S$. As $k_2$ increases, the no-trade region increases and the trading size decreases. This agrees with the economic intuition that higher proportional costs discourage large adjustments.
\begin{figure}[htbp]
	\centering
		\subfloat[]{\label{f-a}
				\includegraphics[width=3in,page={1}]{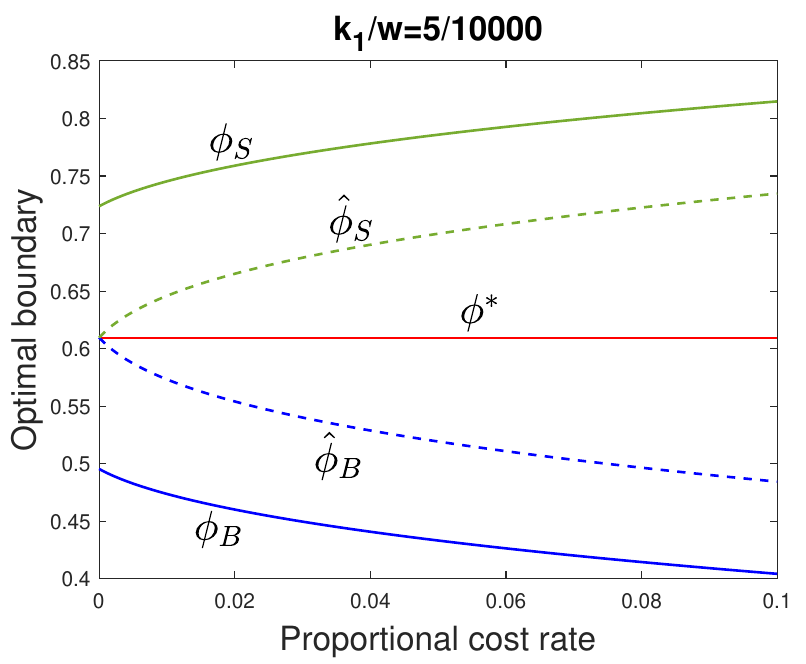}}
		\quad
		\subfloat[]{\label{f-b} 		
			\includegraphics[width=3in, page={2}]{given_fixed.pdf}}
	\caption{{\bf Given fixed-cost-to-wealth ratio $\frac{k_1}{w}$, four free boundaries as functions of the proportional cost rate $k_2$.} {\small The graph plots the optimal buy boundary $\phi_B$, post-buy boundary $\hat{\phi}_B$, sell boundary $\phi_S$, post-sell boundary $\hat{\phi}_S$ and Merton line $\phi^*$ against the proportional cost rate $k_2$ (i) corresponds to $k_1/w=5/10000$, (ii) corresponds to $k_1/w=15/10000$.}}
\label{given-fix prop}
\end{figure}

Figure \ref{consumption fix prop} illustrates how the consumption rate \eqref{c_rate linear} depends on the fixed-cost-to-wealth ratio $\frac{k_1}{w}$ and the proportional cost rate $k_2$. The figure demonstrates that when either $\frac{k_1}{w}$ or $k_2$ increases, the consumption rate decreases. This is because a larger fraction of wealth is absorbed by transaction costs.
\begin{figure}[htbp]
	\centering
	\includegraphics[width=6in]{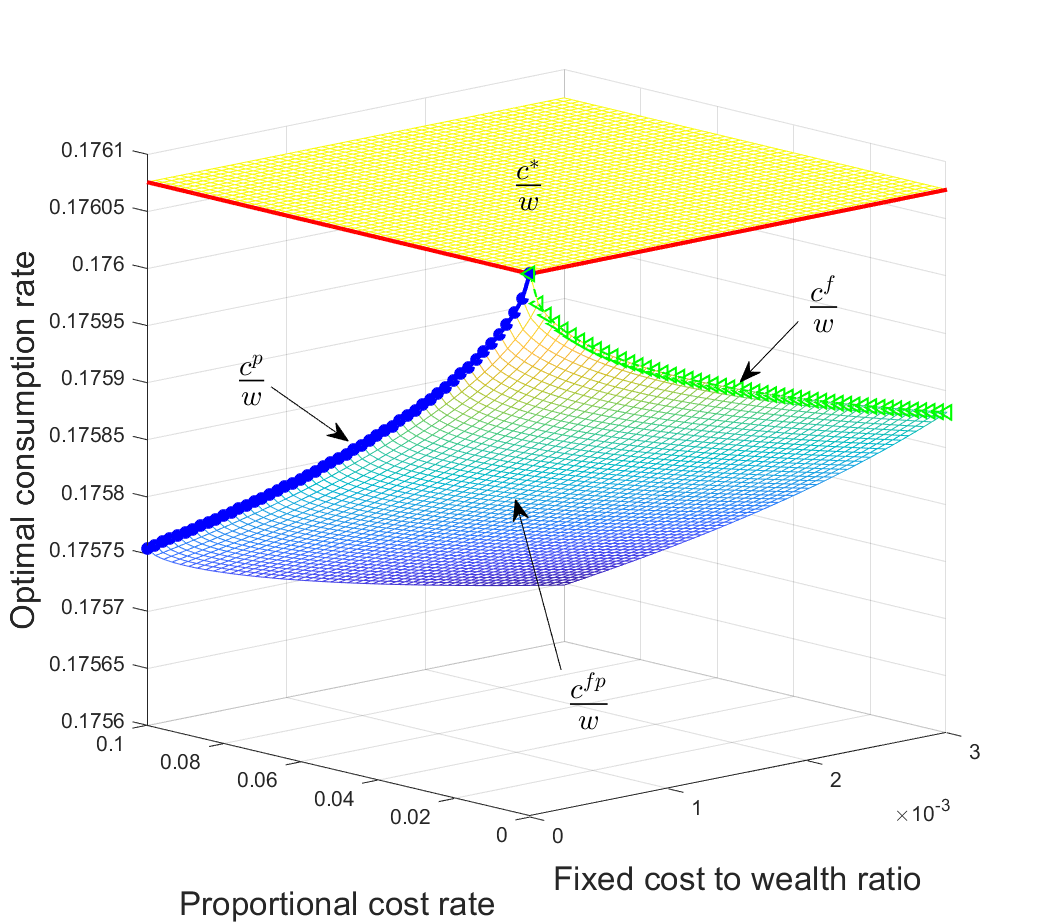}
	\caption{{\bf Optimal consumption rates as functions of fixed and proportional costs.} {\small The graph plots the optimal consumption rate for the linear transaction cost $\frac{c^{\mathrm{fp}}}{w}$, the optimal consumption rate for fixed transaction cost $\frac{c^{\mathrm{f}}}{w}$, the optimal consumption rate for proportional transaction cost $\frac{c^{\mathrm{p}}}{w}$ and the frictionless consumption rate $\frac{c^*}{w}$ against the fixed-cost-to-wealth ratio $\frac{k_1}{w}$ and the proportional cost rate $k_2$.}}
\label{consumption fix prop}
\end{figure}

The following two corollaries verify that the case of fixed cost only \cite{Altarovici-Muhle-Soner-FS-2015, Morton-Pliska-MF-1995} and the case of proportional cost only \cite{Janecek-Shreve-FS-2004} are special cases of Corollary \ref{CRRA linear}.

\begin{Corollary}[Fixed cost only]\label{CRRA fixed}
	A superscript ``${\mathrm{f}}$" denotes quantities associated with the fixed-cost-only case. In this case,
	$k_2=0$ in \eqref{linear costs}, the optimal trading boundaries are 
	\begin{align}
	&\phi^{\mathrm{f}}_S=\phi^*+\left(H \frac {k_1} w \right)^{\frac14},\ \phi^{\mathrm{f}}_B=\phi^*- \left(H \frac {k_1} w \right)^{\frac14},\ \widehat{\phi}^{\mathrm{f}}_S=\widehat{\phi}^{\mathrm{f}}_B=\phi^*,\label{crra optimal boundaries fixed}
	\end{align}
	where $H$ and $\phi^*$ are given by \eqref{general A1 A2} and \eqref{mertoncrra phi c}.
	Moreover, the optimal value function and consumption policies are
	\begin{align}
	V^{\mathrm{f}}(\tau,w)=\,&V^*(\tau,w)\left[1+(1-\gamma)\frac{\widehat{f}^{\mathrm{f}}(\tau,w)}{f^{*}(\tau)}\right],\label{crraleadingV fixed}\\
	c^{\mathrm{f}}(\tau,w)=\,&c^*(\tau,w)+2 a_1(2\gamma - 1)\sigma^2\left(H\frac{k_1}{w}\right)^{\frac12}\frac{\psi(\tau)w}{[f^*(\tau)]^{2}}.\label{c_rate fixed}
	\end{align}
	
	In \eqref{crraleadingV fixed},
	\begin{align}
	&\widehat{f}^{\mathrm{f}}(\tau,w)=- 4 {\gamma\sigma^2} \psi(\tau)\left(H\frac {k_1}{w}\right)^{\frac 1 2 },\label{crraf2 fixed}
	\end{align}
	where $\psi(\tau)$ is given by \eqref{hatI}, and $a_1$ is defined in \eqref{def_a}. 
	$c^*$, $V^*$ and $f^*$ are given by \eqref{mertoncrra phi c}, \eqref{mertoncrra V} and \eqref{mertoncrra f}, respectively. 	
\end{Corollary}
\begin{proof}
	This follows from Corollary~\ref{CRRA linear} by taking $k_2=0$.
\end{proof}
In this case, the post-trade boundaries coincide with the Merton proportion, $\widehat{\phi}^{\mathrm{f}}_B=\widehat{\phi}^{\mathrm{f}}_S=\phi^*$, reflecting the fact that the only friction is a fixed cost per trade, which penalizes frequency but not volume. The formal connection of our asymptotic analysis to the pure fixed-cost models of \cite{Morton-Pliska-MF-1995} and \cite{Altarovici-Muhle-Soner-FS-2015} is therefore established.

\begin{Corollary}[Proportional cost only]\label{CRRA proportional}
	A superscript ``$\mathrm{p}$" denotes quantities associated with the proportional-cost-only case.
	For this case, $k_1=0$ in \eqref{linear costs}, the optimal trading boundaries are 
	\begin{align}
	&\phi^{\mathrm{p}}_S=\widehat{\phi}^{\mathrm{p}}_S=\phi^*+\frac{1}{2}(H k_2)^{\frac 13},\quad
	\phi^{\mathrm{p}}_B=\widehat{\phi}^{\mathrm{p}}_B=\phi^*-\frac{1}{2}(H k_2)^{\frac 13},\label{crra optimal boundaries prop}
	\end{align}
	where $H$ and $\phi^*$ are given by \eqref{general A1 A2} and \eqref{mertoncrra phi c}.
	
	Moreover, the corresponding optimal value function and consumption policy are
	\begin{align}
	V^{\mathrm{p}}(\tau,w)=\,&V^*(\tau,w)\left[1+(1-\gamma)\frac{\widehat{f}^{\mathrm{p}}(\tau,w)}{f^{*}(\tau)}\right],\label{crraleadingV prop}\\
	c^{\mathrm{p}}(\tau,w)=\,&c^*(\tau,w)+ 3 a_1 \gamma\sigma^2\left(Hk_2\right)^{\frac23}\frac{\psi(\tau)w}{[f^*(\tau)]^{2}}.\label{c_rate prop}		
	\end{align}
	In \eqref{crraleadingV prop},
	\begin{align}\label{crraf2 prop}
	\widehat{f}^{\mathrm{p}}(\tau)=-3 {\gamma\sigma^2}\psi(\tau) (H k_2)^{\frac 23},
	\end{align}
	where $\psi(\tau)$ is given by \eqref{hatI}, and $a_1$ is defined in \eqref{def_a}. 
	$c^*$, $V^*$ and $f^*$ are given by \eqref{mertoncrra phi c}, \eqref{mertoncrra V} and \eqref{mertoncrra f}, respectively.	
\end{Corollary}
\begin{proof}
	In this case $k_1 = 0$, and the argument used to derive Corollary~\ref{CRRA linear} is no longer directly applicable because the system \eqref{original generalz} degenerates. This is the pure proportional cost regime, which is well documented in the literature (see, e.g., \cite{Janecek-Shreve-FS-2004} for the case without consumption). The results \eqref{crra optimal boundaries prop}, \eqref{c_rate prop}, and \eqref{crraf2 prop} follow directly from Theorem~\ref{General} by noting that $A_1=0$ and $A_2=Hk_2$ in \eqref{original generalx}, which implies $x^{\mathrm{p}}=0$ and $\widetilde{x}^{\mathrm{p}}=(Hk_2)^{1/3}$ (cf. the case \ref{A1 is 0} in Remark \ref{Remark special cases of original generalx}).
\end{proof}

\begin{Remark}
	The boundary conditions involving only proportional costs differ from the other cases discussed earlier. This is because, in this case, $\phi_B=\widehat{\phi}_B$ and $\phi_S=\widehat{\phi}_S$, which imposes a constraint on the second derivative of the optimal free boundaries $\phi_B$ and $\phi_S$. See Appendix \ref{appendix: proof fixed or prop}.
\end{Remark}

The results $\widehat{\phi}_B= \phi_B$ and $\widehat{\phi}_S= \phi_S$ in the case of only proportional transaction costs have a natural interpretation. When the portfolio is positioned at the trading boundaries $\phi_B$ or $\phi_S$, there are two possibilities in the next moment: (1) the portfolio remains at the trading boundary or moves toward the interior of the no-trade region, in which case no trading is needed; or (2) the portfolio moves outside the no-trade region. In the latter case, one only needs to trade the smallest possible amount of the risky asset to return the portfolio to the trading boundary, thereby minimizing transaction costs. If one instead trades into the interior of the no-trade region, part of the proportional transaction cost may be wasted, since the portfolio may subsequently move back into the no-trade region.

\section{Optimal policies for nonlinear transaction costs}\label{sec: NonLinear costs}

In this section, we explain how the leading-order framework accommodates nonlinear transaction costs and clarify the role played by the fixed and proportional leading components.
Motivated by the nonlinear price-impact literature \cite{Garleanu-Pedersen-JF-2013, Moreau-Muhle-Karbe-Soner-MF-2017, Caye-Herdegen-Muhle-Karbe-AAP-2020, Guasoni-Weber-MF-2020, Almgren-Thum-Risk-2005}, we consider, as a special example, a power-law total transaction cost
$k(\Delta)=k_3\Delta^{\kappa}$ with $\kappa>1$, where $\Delta$ is the dollar amount traded. This form includes the total-cost exponents associated with several commonly studied price-impact and cost specifications: a square-root price impact corresponds to $\kappa=3/2$ \cite{Caye-Herdegen-Muhle-Karbe-AAP-2020, Guasoni-Weber-MF-2020}, a $3/5$ price-impact law corresponds to $\kappa=8/5$ \cite{Almgren-Thum-Risk-2005}, and quadratic trading costs correspond to $\kappa=2$ \cite{Garleanu-Pedersen-JF-2013, Moreau-Muhle-Karbe-Soner-MF-2017}. When fixed and proportional components are absent, these nonlinear specifications share a common feature: the transaction cost function satisfies the conditions
\begin{align}\label{add-condi}
	k(0)=0 \quad\text{and}\quad k'(0+)=0.
\end{align}
For the power-law form, and more generally for any convex cost function covered by Assumption~\ref{ass:cost} that satisfies \eqref{add-condi}, the free-boundary system \eqref{original generalx}--\eqref{general A1 A2} cannot have a positive nonnegative solution. Indeed, convexity and \(k(0)=0\) imply \(k(z)\le z k'(z)\) for \(z>0\). Hence
\[
A_1(x,w)=H\left[\frac{k(xw)}{w}-xk'(xw)\right]\le0.
\]
For a strictly nonlinear power law with \(\kappa>1\), this inequality is strict for every \(x>0\). Since the second equation in \eqref{original generalx} requires \(x^3\widetilde{x}=A_1(x,w)\) with \(x,\widetilde{x}\ge0\), no positive solution exists. The only admissible nonnegative solution is therefore \(x=\widetilde{x}=0\). Thus, the no-trade region collapses, portfolio adjustment becomes continuous, and the limiting policy is consistent with continuous adjustment along the Merton line.

More generally, the analysis shows that the existence of a non-degenerate no-trade region is determined by the fixed or proportional leading component of the transaction cost function. Only when these leading components are absent does the purely nonlinear component lead to the degenerate no-trade case. This can be seen from the following transaction cost specification:
\begin{align}\label{nonlinear costs}
	k(\triangle\phi\,w)=k_1+k_2(\triangle\phi\,w)+k_{\mathrm{nl}}(\triangle\phi\,w),
\end{align}
where $k_1$ and $k_2$ are not both zero, and $k_{\mathrm{nl}}(\triangle\phi\,w)$ is the nonlinear component of the transaction cost. The fixed or proportional leading component generates a non-degenerate no-trade region. If the nonlinear component is retained in the implicit quantities \(A_1(x,w)\) and \(A_2(x,w)\), the selected solution must still satisfy the nonnegativity requirement in \eqref{original generalx}; in particular \(A_1(x,w)\ge0\) at the selected trade size. If a convex nonlinear term is the only leading component, this condition fails except at \(x=0\), as shown above.

To make the discussion more concrete, we consider the case in which the nonlinear component has a power-law form, namely,
\begin{align}\label{quadratic costs}
	k(xw)=k_1+k_2(xw)+k_3(xw)^{\kappa},
\end{align}
with $\kappa>1$. Applying Theorem~\ref{General} yields the following corollary.
\begin{Corollary}[Linear and power-law transaction costs]\label{CRRA quadratic}
	A superscript ``$\mathrm{lp}$'' denotes quantities associated with the linear-plus-power-law cost structure \eqref{quadratic costs}, allowing for fixed, proportional, and power-law components. The fixed component is covered by the right-limit interpretation in Assumption~\ref{ass:cost}. Assume that the implicit system \eqref{original generalx} with the quantities below admits a selected nonnegative solution \((x,\widetilde{x})\). This requirement rules out the pure power-law case \(k_1=k_2=0\), for which the only admissible solution is the degenerate one \(x=\widetilde{x}=0\). The leading-order trading boundaries are
		\begin{align}
			&\phi_S^{\mathrm{lp}}=\phi^*+\tfrac{1}{2}\left(\widetilde{x}+x\right),\  \widehat{\phi}^{\mathrm{lp}}_S=\phi^*+\tfrac{1}{2}\left(\widetilde{x}-x\right),\label{crra optimal sell boundary quadratic}\\
			&\phi_B^{\mathrm{lp}}=\phi^*-\tfrac{1}{2}\left(\widetilde{x}+x\right),\  \widehat{\phi}_B^{\mathrm{lp}}=\phi^*-\tfrac{1}{2}\left(\widetilde{x}-x\right).\label{crra optimal buy boundary quadratic}
		\end{align}
			Here $x$ and $\widetilde{x}$ are determined by \eqref{original generalx} with  
		\begin{equation}\label{general quadratic A1 A2}
			A_1(x,w)=H\Big(\frac{k_1}{w}+(1-\kappa)k_3x(xw)^{\kappa-1}\Big),\quad  A_2(x,w)=H\big(k_2+\kappa k_3(xw)^{\kappa-1}\big),
		\end{equation} 
		and $H$ is defined in \eqref{general A1 A2}. 
		
		The optimal value function and consumption rate are
		\begin{align}
			V^{\mathrm{lp}}(\tau,w)=\,&V^*(\tau,w)\left[1+(1-\gamma)\frac{\widehat{f}^{\mathrm{lp}}(\tau,w)}{f^{*}(\tau)}\right],
			\label{crraleadingV quadratic}
			\\
			c^{\mathrm{lp}}(\tau,w)=\,&c^*(\tau,w)-\frac{a_1w}{[f^*(\tau)]^{2}}\left[\widehat{f}^{\mathrm{lp}}(\tau,w)+\frac{1}{\gamma}w\,\partial_{2}\widehat{f}^{\mathrm{lp}}(\tau,w)\right].\label{c_rate quadratic}
		\end{align}
			In \eqref{crraleadingV quadratic} and \eqref{c_rate quadratic}, $\widehat{f}^{\mathrm{lp}}(\tau,w)$ is given by \eqref{crraV2}, $\psi(\tau)$ is given by \eqref{hatI}, and $a_1$ is defined in \eqref{def_a}. 
		
		In \eqref{crra optimal sell boundary quadratic}--\eqref{c_rate quadratic}, $V^*$, $f^*$, $\phi^*$, and $c^*$ denote the frictionless solutions given by \eqref{mertoncrra V}, \eqref{mertoncrra f}, and \eqref{mertoncrra phi c}. 
	
\end{Corollary}

In summary, it is the fixed or proportional component---not the nonlinear one---of the transaction cost function that determines whether a no-trade region exists.

\section{A verification theorem}\label{sec: verification}

In this section, we provide a verification theorem for the exact control problem formulated in Section~\ref{sec: finaincal market}. The theorem is separate from the asymptotic expansion derived in Section~\ref{sec: general costs}: it shows that any sufficiently smooth solution of the associated quasi-variational inequality is equal to the value function. It does not, by itself, prove convergence of the leading-order asymptotic formulas in Theorem~\ref{General}; such convergence would require additional error estimates.

As discussed in Remark~\ref{rem:cost-approx}, item~\ref{control regimes}, the exact quasi-variational inequality (QVI) \eqref{verification qvi} features a non-local intervention operator $\mathcal{M}$ and discrete jump summations, which explicitly characterize an \textit{impulse control} problem. 

\begin{Remark}[Attainability of the optimal impulse and verification scope]\label{rem:impulse_wellposed}
	For this impulse-control problem to be mathematically well-posed (i.e., to prevent the optimal strategy from degenerating into an absolutely continuous rate control via infinite-order splitting), the transaction cost must exhibit local subadditivity near the origin. Therefore, Theorem~\ref{verification-thm} presented below specifically assumes a strictly positive fixed cost component, $\lim_{z\downarrow 0}k(z) > 0$. The cases of pure proportional costs ($\lim_{z\downarrow 0}k(z)=0, k''=0$) or pure nonlinear costs ($\lim_{z\downarrow 0}k(z)=0, k''>0$) require structurally distinct verification arguments (singular or continuous control, respectively). Although a full technical exposition of these separate arguments is outside the scope of this theorem, the arguments are connected through the general asymptotic expansion framework developed in this paper, and their interpretation as limit transitions of the exact QVI is elaborated later in Remark~\ref{rem:limit_transitions}.
\end{Remark}

Let $x$ and $y$ denote the dollar amounts invested in the risk-free and risky assets, respectively. For a smooth test function $\Phi=\Phi(t,x,y)$, define the no-trade generator under consumption rate $c\ge0$ by
\begin{align}\label{verification generator}
	\mathcal{L}^c\Phi
	=\partial_t\Phi+(rx-c)\partial_x\Phi+\mu y\,\partial_y\Phi
	+\frac12\sigma^2y^2\partial_{yy}\Phi .
\end{align}
For a purchase of dollar amount $\ell>0$ and a sale of dollar amount $j>0$, define the post-trade states
\begin{align}\label{verification posttrade}
	\Gamma_B^\ell(x,y)=(x-\ell-k(\ell),\,y+\ell),\qquad
	\Gamma_S^j(x,y)=(x+j-k(j),\,y-j).
\end{align}
Let $\mathcal{I}_B(x,y)$ be the set of feasible positive purchases for which the corresponding post-trade state remains in $\overline{\mathcal S}$, and let $\mathcal{I}_S(x,y)$ be the set of feasible positive sales satisfying $0<j\le y$ and preserving solvency. The intervention operator is
\begin{align}\label{verification intervention}
	\mathcal{M}\Phi(t,x,y)
	=\max\left\{
	\sup_{\ell\in\mathcal{I}_B(x,y)}\Phi(t,\Gamma_B^\ell(x,y)),\,
	\sup_{j\in\mathcal{I}_S(x,y)}\Phi(t,\Gamma_S^j(x,y))
	\right\},
\end{align}
where the supremum over an empty set is interpreted as $-\infty$. Thus $\mathcal{M}\Phi=-\infty$ if no positive feasible trade is available. The quasi-variational inequality associated with the exact impulse-control problem is
\begin{align}\label{verification qvi}
	\max\left\{
	\sup_{c\ge0}\left[\alpha e^{-\beta t}u_1(c)+\mathcal{L}^c\Phi(t,x,y)\right],
	\mathcal{M}\Phi(t,x,y)-\Phi(t,x,y)
	\right\}=0
\end{align}
on $[0,T)\times\mathcal S$, together with the terminal condition
\begin{align}\label{verification terminal}
	\Phi(T,x,y)=(1-\alpha)e^{-\beta T}u_2(x+y).
\end{align}

\begin{Theorem}[Verification Theorem]\label{verification-thm}
	Suppose that {the transaction cost has a strictly positive fixed component, $\lim_{z\downarrow 0}k(z) > 0$,} and that a function $\Phi:[0,T]\times\overline{\mathcal{S}}\to\mathbb{R}$ satisfies the following conditions:
	\begin{enumerate}[label=(\roman*)]
		\item\label{regularity} \textbf{Regularity}: {$\Phi$ is continuous on $[0,T]\times\overline{\mathcal{S}}$ and of class $C^1$ on $[0,T)\times\mathcal{S}$. For each $t \in [0,T)$, the spatial function $\Phi(t,\cdot,\cdot)$ belongs to the local Sobolev space $W^{2,\infty}_{loc}(\mathcal{S})$. Furthermore,} $\Phi$ is of class {$C^{1,2}$} in the interiors of the continuation and intervention regions. 
		
		\item\label{terminal} \textbf{Terminal condition}: $\Phi$ satisfies the terminal condition \eqref{verification terminal}.
		
		\item\label{QVI compliance} \textbf{QVI compliance}: $\Phi$ satisfies the QVI \eqref{verification qvi}. {Specifically, the intervention obstacle condition $\Phi(t,x,y) \ge \mathcal{M}\Phi(t,x,y)$ holds everywhere on $[0,T)\times\mathcal{S}$, and the differential inequality $\sup_{c\ge0}\big[\alpha e^{-\beta t}u_1(c)+\mathcal{L}^c\Phi(t,x,y)\big] \le 0$ holds Lebesgue-almost everywhere on $[0,T)\times\mathcal{S}$.}
		
		\item\label{UI} \textbf{Integrability and occupation time}: {For any admissible strategy $\pi \in \mathcal{A}(t,x,y)$, the family of random variables
			\[
			\big\{ \Phi(\tau, w^B(\tau),w^S(\tau)) \big\}_{\tau \in \mathcal{T}_{[t,T]}}
			\]
		is uniformly integrable, where $\mathcal{T}_{[t,T]}$ denotes the set of all stopping times taking values in $[t,T]$. Additionally, the expected occupation time of the corresponding controlled state process on the Lebesgue-null set where $\Phi$ is not of class $C^{1,2}$ is zero.}
				
	\end{enumerate}
	Then
	\[
	\Phi(t,x,y)\ge V(t,x,y),\qquad \text{for all } (t,x,y)\in[0,T]\times\mathcal{S}.
	\]
	Moreover, suppose there exists an admissible feedback strategy $\pi^*=(c^*,\mathfrak L^*,\mathfrak J^*)$ such that:
	\begin{enumerate}[label=(\roman*),resume]
		\item In the continuation region, $c^*$ attains the supremum in the HJB part of \eqref{verification qvi}.
		\item At each trading time $\tau_i$, the chosen impulse attains the intervention value $\mathcal M\Phi$, such that the post-trade state satisfies
		\[
		\Phi(\tau_i,w^B(\tau_i-),w^S(\tau_i-))
		=\Phi(\tau_i,w^B(\tau_i),w^S(\tau_i)).
		\]
		\item The strategy $\pi^*$ is admissible, and {condition \ref{UI} holds under $\pi^*$.}
	\end{enumerate}
	Then $\Phi(t,x,y)=V(t,x,y)$, and $\pi^*$ is an optimal strategy.
\end{Theorem}

\begin{proof}
	Fix an arbitrary admissible strategy $\pi=(c,\mathfrak L,\mathfrak J)$, and let $(w^B(s),w^S(s))$ be the corresponding controlled state. Let $\{\rho_n\}_{n\ge 1}$ be a localizing sequence of stopping times such that $\rho_n \to \infty$ almost surely, the stopped state remains in a compact subset of $\mathcal S$ between trading times, and the stopped stochastic integrals are true martingales.
	
	{Since $\Phi(t,\cdot,\cdot) \in W^{2,\infty}_{loc}(\mathcal{S})$, its first spatial derivatives are locally Lipschitz continuous, implying that its second distributional derivatives contain no singular measure components. Let $\mathcal{N} \subset [0,T) \times \mathcal{S}$ denote the Lebesgue-null set where $\Phi$ is not strictly of class $C^{1,2}$. By the zero occupation time assumption in condition \ref{UI}, we have
		\[
		\mathbb{E}\left[ \int_t^T \mathbf{1}_{\mathcal{N}}(s, w^B(s), w^S(s)) \, ds \right] = 0.
		\]
		Consequently, the state process almost surely accumulates no local time on $\mathcal{N}$, allowing us to apply the generalized It\^o formula (e.g., Meyer-It\^o formula) to obtain that,} on each interval $[\tau_i, \tau_{i+1} \wedge \rho_n)$, 
	\[
	d\Phi(s,w^B(s),w^S(s))
	=\mathcal L^{c(s)}\Phi(s,w^B(s),w^S(s))\,ds
	+\sigma w^S(s)\partial_y\Phi(s,w^B(s),w^S(s))\,dW(s). 
	\]
	
	By condition \ref{QVI compliance}, $\alpha e^{-\beta s}u_1(c(s))+\mathcal L^{c(s)}\Phi(s,w^B(s),w^S(s))\le0$ Lebesgue-almost everywhere. Therefore, between trades,
	\[
	d\Phi(s,w^B(s),w^S(s))
	\le -\alpha e^{-\beta s}u_1(c(s))\,ds
	+\sigma w^S(s)\partial_y\Phi(s,w^B(s),w^S(s))\,dW(s) .
	\]
	At a trading time $\tau_i$, the intervention condition $\Phi\ge\mathcal M\Phi$ guarantees that the value function candidate does not strictly increase across jumps:
	\[
	\Phi(\tau_i,w^B(\tau_i),w^S(\tau_i))
	\le
	\Phi(\tau_i,w^B(\tau_i-),w^S(\tau_i-)).
	\]
	
	{Furthermore, since the transaction cost has a strictly positive fixed component ($\lim_{z\downarrow 0}k(z) > 0$) and the admissible strategy must satisfy the solvency condition, the number of trading times in any finite interval $[0,T]$ is almost surely finite. This rigorously precludes Zeno behavior and ensures that the sum over jump interventions is mathematically well-defined and finite.}
	
	Summing the generalized It\^o inequalities over all no-trade intervals up to $T\wedge\rho_n$, adding the nonpositive jump contributions, and taking expectations yield
	\[
	\Phi(t,x,y)
	\ge
	\mathbb E\left[
	\int_t^{T\wedge\rho_n}\alpha e^{-\beta s}u_1(c(s))\,ds
	+\Phi(T\wedge\rho_n,w^B(T\wedge\rho_n),w^S(T\wedge\rho_n))
	\right].
	\]
	Letting $n\to\infty$, {we have \(T \wedge \rho_n \to T\) almost surely. The uniform integrability assumed in condition \ref{UI} ensures convergence in $L^1$ for the terminal term. For the running cost integral, the integrability condition in Definition \ref{def:admissible}~(\ref{Integrability}) guarantees \(\mathbb{E}\left[ \int_t^T e^{-\beta s} |u_1(c(s))| \, ds \right] < \infty\), permitting the direct application of the Dominated Convergence Theorem. Passing to the limit yields}
	\[
	\Phi(t,x,y)
	\ge
	\mathbb E\left[
	\int_t^T\alpha e^{-\beta s}u_1(c(s))\,ds
	+(1-\alpha)e^{-\beta T}u_2(w^B(T)+w^S(T))
	\right].
	\]
	Since the admissible strategy $\pi$ was arbitrary, $\Phi(t,x,y)\ge V(t,x,y)$.
	
	Now suppose the feedback strategy $\pi^*$ satisfies the specified optimality conditions. Along the no-trade intervals, the HJB supremum is attained, making the drift inequality an equality. At trading times, the selected intervention attains $\mathcal M\Phi$, making the jump inequality an equality. Applying the identical localization and convergence arguments under \(\pi^*\) gives
	\[
	\Phi(t,x,y)
	=
	\mathbb E\left[
	\int_t^T\alpha e^{-\beta s}u_1(c^*(s))\,ds
	+(1-\alpha)e^{-\beta T}u_2(w^{B,*}(T)+w^{S,*}(T))
	\right].
	\]
	Hence $V(t,x,y)\ge\Phi(t,x,y)$. Combining this with the upper bound proves $V=\Phi$, and the feedback strategy $\pi^*$ is optimal.
\end{proof}

\begin{Remark}
	In the continuation region $\Phi>\mathcal M\Phi$, the QVI reduces to the HJB equation. After the change of variables $w=x+y$ and $\phi=y/w$, this corresponds to \eqref{HJBVcrra} written in the original variables. The buy and sell boundaries in Section~\ref{sec: general costs} correspond to points where $\Phi=\mathcal M\Phi$. Value matching and first-order conditions are obtained by evaluating the intervention equality and differentiating with respect to the optimal trade size, provided the stated differentiability assumptions hold.
\end{Remark}

\begin{Remark}[Limit Transitions to Degenerate Control Regimes]\label{rem:limit_transitions}
	While Theorem~\ref{verification-thm} is strictly formulated for the impulse-control regime (where $\lim_{z\downarrow 0}k(z) > 0$), our general asymptotic expansion framework implicitly unifies other control regimes. These regimes can be rigorously interpreted as limit transitions of the exact QVI \eqref{verification qvi} when the fixed cost component vanishes. The nature of the topological phase transition in the space of admissible strategies depends fundamentally on the local marginal cost near the origin, $\lim_{z\downarrow 0}k'(z)$:
	
	\begin{enumerate}[label=(\alph*)]
		\item \textbf{Singular Control Limit ($\lim_{z\downarrow 0}k'(z) > 0$)}: 
		Assume there is no fixed cost ($\lim_{z\downarrow 0}k(z) = 0$), but the marginal cost at the origin is strictly positive, denoted by $k'(0+) \triangleq \lim_{z\downarrow 0}k'(z) > 0$. The basic example is pure proportional costs ($k(z) = k_2 z$). As the trade size $z \to 0$, the optimal strategy shifts from discrete jumps to continuous, finite-variation processes. Analytically, applying a first-order Taylor expansion to the intervention obstacle $\Phi(x,y) \ge \mathcal{M}\Phi(x,y)$ for an infinitesimal purchase $\ell > 0$ yields:
		\[
		\Phi(x,y) \ge \Phi(x - \ell - k(\ell), y + \ell) \approx \Phi(x,y) + \big[ \partial_y\Phi - \big(1+k'(0+)\big)\partial_x\Phi \big] \ell.
		\]
		Dividing by $\ell$ and letting $\ell \downarrow 0$ recovers the classical gradient constraint $\partial_y\Phi - (1+k'(0+))\partial_x\Phi \le 0$ for the buy region. A symmetric argument for an infinitesimal sale yields $(1-k'(0+))\partial_x\Phi - \partial_y\Phi \le 0$ for the sell region. The verification argument adapts by replacing the discrete jump summations in the Meyer-It\^o formula with Lebesgue-Stieltjes integrals with respect to the finite-variation processes $\mathfrak{L}$ and $\mathfrak{J}$.
		
		\item \textbf{Absolutely Continuous Control Limit ($\lim_{z\downarrow 0}k'(z) = 0$)}: Assume there is no fixed cost ( $\lim_{z\downarrow 0}k(z) = 0$) and the marginal cost vanishes at the origin ($k'(0+) = 0$), but the cost function is strictly convex ($k''(z) > 0$). The basic example is pure nonlinear costs ($k(z) = k_3 z^\kappa$ for $\kappa>1$). As discussed in Remark~\ref{rem:cost-approx}, item~\ref{control regimes}, infinite-order splitting reduces the absolute cost of any instantaneous block trade to zero. Consequently, the rigorous formulation inherently shifts to an absolutely continuous control regime, where the investor trades at finite rates $l(t) = d\mathfrak{L}(t)/dt \ge 0$ and $j(t) = d\mathfrak{J}(t)/dt \ge 0$ over time. The non-local intervention operator in the QVI is replaced by a local supremum over consumption and trading rates within the continuous HJB equation:
		\[
		\sup_{c \ge 0, \, l, j \ge 0} \left\{ \alpha e^{-\beta t}u_1(c) + \mathcal{L}^c\Phi + l\big[\partial_y\Phi - \partial_x\Phi\big] + j\big[\partial_x\Phi - \partial_y\Phi\big] - \tilde{k}(l) \partial_x\Phi - \tilde{k}(j) \partial_x\Phi \right\} = 0.
		\]
		Here, $\tilde{k}(\cdot)$ represents the instantaneous execution cost \textit{rate} (e.g., $\tilde{k}(l) = \text{const} \cdot l^\kappa$), which is dimensionally distinct from the normalized block-trade cost $\widehat{k}(\cdot,\cdot)$ defined earlier. It serves as the continuous-time flow analogue of $k(\cdot)$. The verification argument in this regime becomes that of a standard continuous-control problem, relying strictly on absolutely continuous Lebesgue integrals.
	\end{enumerate}
	These limit transitions clarify how the verification arguments differ across transaction-cost specifications while remaining consistent with the same solvency-based economic structure.
\end{Remark}

\begin{Remark}
	Theorem~\ref{verification-thm} verifies a sufficiently smooth exact solution of the QVI. To verify the leading-order asymptotic formulas in Theorem~\ref{General} as asymptotically optimal policies, one would need additional residual estimates showing that the asymptotic candidate satisfies the QVI up to a controlled error and that this error vanishes as the cost-to-wealth ratio tends to zero.
\end{Remark}

\section{Conclusion}\label{sec: conclusion}

In this paper, we study optimal asset allocation and consumption strategies under a broad class of transaction cost functions satisfying smoothness, monotonicity, and convexity conditions. We derive leading-order asymptotic formulas for the transaction-cost-induced corrections, including the no-trade region, the four trading boundaries, the value-function correction, and the optimal consumption rate. These formulas address the practical questions of when to trade, how much to trade, how to allocate wealth, and how to consume. Our theoretical approach is based on maximizing expected CRRA utility over a finite horizon and applying a singular perturbation expansion. We show that the fixed and proportional components determine the leading-order size of the no-trade region and characterize analytically how the trading boundaries and trading volumes depend on the cost structure.

Complementing the asymptotic analysis, we establish a rigorous verification theorem for the exact formulation of the optimal control problem. Specifically, we prove that under a strictly positive fixed cost component, a sufficiently regular solution to the associated impulse-control quasi-variational inequality (QVI) coincides with the true value function. Furthermore, we analyze the limiting transitions of this exact QVI to degenerate control regimes---namely, singular and absolutely continuous control---as the fixed cost vanishes.

These results provide analytical guidance for portfolio and risk management in markets with transaction costs. The analysis also clarifies the roles of linear and nonlinear cost components, showing that the existence of a no-trade region is governed by the fixed and proportional components rather than by the nonlinear component alone. Overall, the proposed framework accommodates general nonlinear transaction cost structures while retaining tractable leading-order formulas for portfolio allocation, consumption, and trading decisions.

An important contribution of this paper is the combination of a general transaction cost function with CRRA preferences. Unlike exponential-utility formulations, where trading decisions are naturally expressed in dollar amounts, the CRRA formulation bases investment decisions on the ratio of risky wealth to total wealth. Accordingly, the relevant {small cost-to-wealth condition} is that the wealth-normalized cost $k(\Delta)/w$ is small, rather than the absolute transaction cost $k(\Delta)$ itself. This unified CRRA framework accommodates fixed, proportional, fixed-plus-proportional, and nonlinear cost structures while providing tractable analytic expressions for the trading boundaries, asset allocation, value-function correction, and optimal consumption adjustment; it also identifies the fixed or proportional leading component as the source of a persistent non-degenerate no-trade region.

\section*{Appendix} 
\appendix

\section{Proof of Theorem \ref{General}}\label{appendix: proof general}

In this appendix, we provide the details for proving Theorem \ref{General}. We define a dimensionless transaction cost $\widehat{k}\big(\triangle\phi,w\big)= \frac{k\big(\triangle\phi\,w\big)}{w}$. This quantity represents the fraction of total wealth spent on the transaction if a trade is performed. In practice, $\widehat k\big(\triangle\phi,w\big)\ll 1$, otherwise, the trade would deplete too much wealth, and the investor will refrain from trading at that moment.

Let $\varepsilon$ ($\varepsilon \ll 1$) be the order of magnitude of the dimensionless transaction cost $\widehat k\big(\triangle\phi,w\big)$. Then 
$\overline{k}\big(\triangle\phi,w\big)= \frac{\widehat k\big(\triangle\phi,\,w\big)}{\varepsilon}$ is ${\rm O}(1)$ term.
Therefore, the transaction cost $k\big(\triangle\phi\,w\big)$ can be expressed as	
\begin{align*}
k\big(\triangle\phi\,w\big)=\varepsilon \overline k\big(\triangle\phi,w\big)w.
\end{align*}
Since $\varepsilon$ is small, the width of the no-trade region is also small.  We expand $\phi$ in terms of $\varepsilon$ around $\phi^*$. To obtain the solution to \eqref{crraH-J-BVc*} and \eqref{crrabc1}-\eqref{crrasc3}, we apply a singular perturbation expansion in terms of $\varepsilon$. 

Let $\lambda$ be the exponent of the leading-order term in the singular perturbation expansion, 
namely  $\phi-\phi^*={\rm O}(\varepsilon^{\lambda})$, where $\lambda$ is a constant to be determined, and $\phi^*$ is the frictionless optimal asset allocation strategy given by \eqref{mertoncrra phi c}. This suggests introducing the scaled variable
\begin{equation}\label{generalYeta}
Y\triangleq \varepsilon^{-\lambda}(\phi-\phi^*),
\end{equation}
and $Y$ is an ${\rm O}(1)$ term. Note that $\varepsilon^{\lambda}$ is the size of the no-trade region. The transaction cost $k\big(\triangle\phi\,w\big)$ can be rewritten as
\begin{align}\label{general cost}
k\big(\triangle\phi\, w\big)
=\varepsilon\overline{k}\big(\triangle\phi,w\big)w
=\varepsilon K\big(\triangle Y,w\big)w.
\end{align}
We change the state variables of the system from $(\tau,w,\phi)$ to $(\tau,w,Y,\varepsilon^{\lambda})$, namely $f(\tau,w,\phi) = \overline{f}(\tau,w,Y,\varepsilon^{\lambda})$.
Although $\overline{f}$ and $f$ have different functional forms, since we will only examine the solution in terms of $\overline{f}$ in the rest of the paper, we drop the overbar on $f$ for notational conciseness.

We expand $f(\tau,w,Y,\varepsilon^{\lambda})$ in powers of $\varepsilon^{\lambda}$:
\begin{equation}\label{generalfYexpasion}
f(\tau,w,Y, \varepsilon^{\lambda})=f^*(\tau)+\sum_{i=1}^{\infty}\varepsilon^{i\lambda}f_i(\tau,w,Y).
\end{equation}
The initial condition becomes
$
f(0,w,Y,\varepsilon^{\lambda})=f^*(0)+\sum_{i=1}^{\infty}\varepsilon^{i\lambda}f_i(0,w,Y)=(1-\alpha)^{\frac{1}{1-\gamma}},
$
which gives
\begin{equation}\label{generalfinalY}
f^*(0)=(1-\alpha)^{\frac{1}{1-\gamma}}\ \ \text{and}\ \ f_i(0,w,Y)=0,\ \text{for}\ i=1,2,3,\dots.
\end{equation}

We rewrite the H-J-B equation \eqref{crraH-J-BVc*} in terms of $Y$, then the coefficients of the resulting equation will depend on $\varepsilon^{\lambda}$. We substitute the expression \eqref{generalfYexpasion} into the resulting equation and regroup the results in the power of $\varepsilon^{\lambda}$ to arrive at
\begin{align}\label{G}
G_{0}(f^*)&+\varepsilon^{\lambda}G_{1}(f^*,f_1)+\varepsilon^{2\lambda}G_{2}(f^*,f_1,f_2)\notag\\
&+\varepsilon^{3\lambda}G_{3}(f^*,f_1,f_2,f_3)+\varepsilon^{4\lambda}G_{4}(f^*,f_1,f_2,f_3,f_4)+\cdots=0.
\end{align}
This yields $G_j = 0$  for $j=0,1,2,3,4,\dots$. The explicit expressions of $G_j$ will be given by \eqref{G equals CD} in Appendix \ref{appendix: details in no-trade}. The functions $G_j$, $j=0,1,2,\dots$, are independent of $\varepsilon$.

We also express the four critical boundaries $\phi_B$, $\phi_S$, $\widehat{\phi}_B$, $\widehat{\phi}_S$ in terms of $Y$, namely 
\begin{equation}\label{generalYeta boundary}
Y_B\triangleq \varepsilon^{-\lambda}(\phi_B-\phi^*),\quad
\widehat{Y}_B\triangleq \varepsilon^{-\lambda}(\widehat{\phi}_B-\phi^*),\quad
Y_S\triangleq \varepsilon^{-\lambda}(\phi_S-\phi^*),\quad
\widehat{Y}_S\triangleq \varepsilon^{-\lambda}(\widehat{\phi}_S-\phi^*).
\end{equation}
Then boundary conditions \eqref{crrabc1}-\eqref{crrasc3} can be expressed in terms of $Y_B$, $\widehat{Y}_B$, $Y_S$ and $\widehat{Y}_S$ (details are in Appendix \ref{general B.C.}).
\begin{align}
&f\big(\tau,w,Y_B,\varepsilon^{\lambda}\big)-f\big(\tau,w,\widehat{Y}_B,\varepsilon^{\lambda}\big)
=\,\varepsilon\triangle W_B\left[w\,\partial_{2}f\big(\tau,w,\widehat{Y}_B,\varepsilon^{\lambda}\big)+\frac{\gamma}{1-\gamma}f\big(\tau,w,\widehat{Y}_B,\varepsilon^{\lambda}\big)\right],\label{crraBC1}\\
&\partial_{Y_B}f\big(\tau,w,Y_B,\varepsilon^{\lambda}\big)
=\,\varepsilon\,\partial_{1}K\big(\triangle Y_B,w\big)
\left[\frac{\gamma}{1-\gamma}f\big(\tau,w,Y_B,\varepsilon^{\lambda}\big)+w\partial_{2}f\big(\tau,w,\widehat{Y}_B,\varepsilon^{\lambda}\big)\right],\label{crraBC2}
\\
&\partial_{\widehat{Y}_B}f\big(\tau,w,\widehat{Y}_B,\varepsilon^{\lambda}\big)
=\,\varepsilon\frac{\gamma}{1-\gamma}\,\partial_{1}K\big(\triangle Y_B,w\big)f\big(\tau,w,Y_B,\varepsilon^{\lambda}\big)\notag\\
&\hspace*{100pt}+\varepsilon w\left[\partial_{1}K\big(\triangle Y_B,w\big)\partial_{2}f\big(\tau,w,\widehat{Y}_B,\varepsilon^{\lambda}\big)+K\big(\triangle Y_B,w\big)\partial_{32}f\big(\tau,w,\widehat{Y}_B,\varepsilon^{\lambda}\big)\right],\label{crraBC3}
\\
&f\big(\tau,w,Y_S,\varepsilon^{\lambda}\big)-f\big(\tau,w,\widehat{Y}_S,\varepsilon^{\lambda}\big)
=\,\varepsilon\triangle W_S\left[w\,\partial_{2}f\big(\tau,w,\widehat{Y}_S,\varepsilon^{\lambda}\big)+\frac{\gamma}{1-\gamma}f\big(\tau,w,\widehat{Y}_S,\varepsilon^{\lambda}\big)\right],\label{crraSC1}\\
&\partial_{Y_S}f\big(\tau,w,Y_S,\varepsilon^{\lambda}\big)
=\,-\varepsilon\,\partial_{1}K\big(\triangle Y_S, w\big)
\left[\frac{\gamma}{1-\gamma}f\big(\tau,w,Y_S,\varepsilon^{\lambda}\big)+w\,\partial_{2}f\big(\tau,w,\widehat{Y}_S,\varepsilon^{\lambda}\big)\right],\label{crraSC2}
\\
&\partial_{\widehat{Y}_S}f\big(\tau,w,\widehat{Y}_S,\varepsilon^{\lambda}\big)
=\,-\varepsilon\frac{\gamma}{1-\gamma}\,\partial_{1}K\big(\triangle Y_S,w\big)f\big(\tau,w,Y_S,\varepsilon^{\lambda}\big)\notag\\
&\hspace*{100pt}-\varepsilon w\left[\partial_{1}K\big(\triangle Y_S,w\big)\partial_{2}f\big(\tau,w,\widehat{Y}_S,\varepsilon^{\lambda}\big)-K\big(\triangle Y_S,w\big)\partial_{32}f\big(\tau,w,\widehat{Y}_S,\varepsilon^{\lambda}\big)\right],\label{crraSC3}
\end{align}
where $\triangle W_B=-K\left(\triangle Y_B,w\right)$, $\triangle W_S=-K\left(\triangle Y_S,w\right)$, $\triangle Y_B=\widehat{Y}_B-Y_B$, and $\triangle Y_S=Y_S-\widehat{Y}_S$.

Define the constant, used repeatedly below,
\begin{align}\label{q}
q\triangleq\frac{1}{2}\sigma^2(\phi^{*})^2\left(1-\frac{A}{1-\gamma}\right)^2.
\end{align}

Substituting \eqref{generalfYexpasion} into boundary conditions \eqref{crraBC1}-\eqref{crraSC3} yields
\begin{align*}
&\sum_{i=1}^{\infty}\varepsilon^{i\lambda}\left[f_i\big(\tau,w,Y_B\big)-f_i\big(\tau,w,\widehat{Y}_B\big)\right]\notag\\
=&\,\varepsilon\triangle W_B\left[\frac{\gamma}{1-\gamma}f^*(\tau)+\sum_{i=1}^{\infty}\varepsilon^{i\lambda}\left(w\,\partial_{2}f_{i}\big(\tau,w,\widehat{Y}_B\big)+\frac{\gamma}{1-\gamma}f_i\big(\tau,w,\widehat{Y}_B\big)\right)\right],
\\
&\sum_{i=1}^{\infty}\varepsilon^{i\lambda}\,\partial_{Y_B}f_{i}\big(\tau,w,Y_B\big)\notag\\
=&\,\varepsilon\,\partial_{1}K\big(\triangle Y_B,w\big)
\left[\frac{\gamma}{1-\gamma}f^*(\tau)+\sum_{i=1}^{\infty}\varepsilon^{i\lambda}\left(\frac{\gamma}{1-\gamma}f_i\big(\tau,w,Y_B\big)+w\,\partial_{2}f_{i}\big(\tau,w,\widehat{Y}_B\big)\right)\right],
\\
&\sum_{i=1}^{\infty}\varepsilon^{i\lambda}\,\partial_{\widehat{Y}_B}f_{i}\big(\tau,w,\widehat{Y}_B\big)\notag\\
=&\,\varepsilon\,\partial_{1}K\big(\triangle Y_B,w\big)\left[\frac{\gamma}{1-\gamma}f^*(\tau)+\sum_{i=1}^{\infty}\varepsilon^{i\lambda}\left(\frac{\gamma}{1-\gamma}f_{i}\big(\tau,w,Y_B\big)+w\,\partial_{2}f_{i}\big(\tau,w,\widehat{Y}_B\big)\right)\right]\notag\\
&+\varepsilon K\big(\triangle Y_B,w\big)\sum_{i=1}^{\infty}\varepsilon^{i\lambda}\,\partial_{32}f_{i}\big(\tau,w,\widehat{Y}_B\big),
\end{align*}
\begin{align*}
&\sum_{i=1}^{\infty}\varepsilon^{i\lambda}\left[f_i\big(\tau,w,Y_S\big)-f_i\big(\tau,w,\widehat{Y}_S\big)\right]\notag\\
=&\,\varepsilon\triangle W_S\left[\frac{\gamma}{1-\gamma}f^*(\tau)+\sum_{i=1}^{\infty}\varepsilon^{i\lambda}\left(w\,\partial_{2}f_{i}\big(\tau,w,\widehat{Y}_S\big)+\frac{\gamma}{1-\gamma}f_i\big(\tau,w,\widehat{Y}_S\big)\right)\right],
\\
&\sum_{i=1}^{\infty}\varepsilon^{i\lambda}\,\partial_{Y_S}f_{i}\big(\tau,w,Y_S\big)\\
=&\,-\varepsilon\,\partial_{1}K\big(\triangle Y_S,w\big)
\left[\frac{\gamma}{1-\gamma}f^*(\tau)+\sum_{i=1}^{\infty}\varepsilon^{i\lambda}\left(\frac{\gamma}{1-\gamma}f_i\big(\tau,w,Y_S\big)+w\,\partial_{2}f_{i}\big(\tau,w,\widehat{Y}_S\big)\right)\right],
\end{align*}
and
\begin{align*}
&\sum_{i=1}^{\infty}\varepsilon^{i\lambda}\,\partial_{\widehat{Y}_S}f_{i}\big(\tau,w,\widehat{Y}_S\big)\\
=&\,-\varepsilon\,\partial_{1}K\big(\triangle Y_S,w\big)\left[\frac{\gamma}{1-\gamma}f^*(\tau)+\sum_{i=1}^{\infty}\varepsilon^{i\lambda}\left(\frac{\gamma}{1-\gamma}f_i\big(\tau,w,Y_S\big)+w\,\partial_{2}f_{i}\big(\tau,w,\widehat{Y}_S\big)\right)\right]\notag\\
&+\varepsilon wK\big(\triangle Y_S,w\big)\sum_{i=1}^{\infty}\varepsilon^{i\lambda}\,\partial_{32}f_{i}\big(\tau,w,\widehat{Y}_S\big).
\end{align*}
To determine the leading-order solution, we only need to keep the first term on the right-hand side of the above boundary conditions, namely the term proportional to $f^*(\tau)$. All other terms have orders higher than ${\rm O}(\varepsilon)$ and are negligible. 
We will examine this system order-by-order, namely for the order $\varepsilon^{i\lambda}$($i=0,1,2,3,4,\dots$), we solve the H-J-B equation (see \eqref{G})
\begin{align}\label{i-th HJB}
G_i(f^*,f_1,\dots,f_i) = 0 
\end{align} 
with the associated boundary conditions. When we match the boundary conditions, there are only two possibilities: $i\lambda = 1$ or $i\lambda \neq 1$. We first examine the possibility for $i\lambda=1$, by solving \eqref{i-th HJB} with the initial condition \eqref{generalfinalY} and the associated boundary conditions given by
\begin{equation}\label{boundary condition ilambda=1}
\left\{\begin{aligned}
&f_i\big(\tau,w,Y_B\big)-f_i\big(\tau,w,\widehat{Y}_B\big) = \triangle W_B\frac{\gamma}{1-\gamma}f^*(\tau),\\
&\partial_{3}f_{i}\big(\tau,w,Y_B\big)
=\,\partial_{3}f_{i}\big(\tau,w,\widehat{Y}_B\big)=\,\partial_{1}K\big(\triangle Y_B,w\big)\frac{\gamma}{1-\gamma}f^*(\tau),\\
&f_i\big(\tau,w,Y_S\big)-f_i\big(\tau,w,\widehat{Y}_S\big) = \triangle W_S\frac{\gamma}{1-\gamma}f^*(\tau),\\
&\partial_{3}f_{i}\big(\tau,w,Y_S\big)
=\,\partial_{3}f_{i}\big(\tau,w,\widehat{Y}_S\big)=\,-\partial_{1}K\big(\triangle Y_S,w\big)\frac{\gamma}{1-\gamma}f^*(\tau).
\end{aligned}\right.
\end{equation}
If a solution exists, then we have found the leading-order solution and the value of $\lambda$ is determined, namely $\lambda = 1/i$. No further examination of higher orders is needed. If a solution does not exist, it must be $i\lambda\neq 1$, then we solve \eqref{i-th HJB} with the initial condition \eqref{generalfinalY} and the associated boundary conditions given by
\begin{equation}\label{boundary condition ilambda neq 1}
\left\{\begin{aligned}
&f_i\big(\tau,w,Y_B\big)-f_i\big(\tau,w,\widehat{Y}_B\big) = 0,\\
&\partial_{3}f_{i}\big(\tau,w,Y_B\big)
=\,\partial_{3}f_{i}\big(\tau,w,\widehat{Y}_B\big)=\, 0,\\
&f_i\big(\tau,w,Y_S\big)-f_i\big(\tau,w,\widehat{Y}_S\big) = 0,\\
&\partial_{3}f_{i}\big(\tau,w,Y_S\big)
=\,\partial_{3}f_{i}\big(\tau,w,\widehat{Y}_S\big) = 0.
\end{aligned}\right.
\end{equation}  
Afterwards, we progress to the next order, namely the order $(i+1)\lambda$. 

We start from $i=0$. We now show that we need to carry out this procedure up to $i=4$ to determine the leading-order contribution from the transaction costs.

{\bf (1)} For $i=0$, namely the $\varepsilon^{0}$ term, \eqref{i-th HJB} is $G_{0}(f^*)=qf^*_{YY}=0$.

This equation holds automatically since $f^*(\tau)$ is a function of $\tau$ only.

{\bf (2)} For $i=1$, namely the $\varepsilon^{\lambda}$ term, from the fact that $f^*$ is independent of $Y$, the H-J-B equation \eqref{i-th HJB} with $i=1$ becomes
\begin{align*}
G_{1}(f^*,f_1)=qf_{1YY}=0,
\end{align*}
which gives $f_1(\tau,w,Y)=a(\tau,w)Y+b(\tau,w)$ with $a$, $b$ being functions which only depend on $\tau$ and $w$. If $\lambda = 1$, the boundary conditions \eqref{boundary condition ilambda=1} with $i=1$ leads to
\begin{align*}
a(\tau,w)\triangle Y_B=\,&\frac{\gamma}{1-\gamma}K\big(\triangle Y_B,w\big)f^*(\tau),\\
a(\tau,w) =\,&\frac{\gamma}{1-\gamma}\,\partial_{1}K\big(\triangle Y_B,w\big)f^*(\tau),
\end{align*}
which implies 
\begin{align}\label{K arbitrary}
\partial_{1}K\big(\triangle Y_B,w\big)\triangle Y_B=K\big(\triangle Y_B,w\big).
\end{align}
Since $K\big(\triangle Y_B,w\big)$ is an arbitrary function, the condition \eqref{K arbitrary} does not hold in general, thus $\lambda\neq 1$. Furthermore, from boundary conditions \eqref{boundary condition ilambda neq 1} with $i=1$, we conclude that $a(\tau,w)\equiv 0$, which implies that $f_1$ is independent of $Y$.

{\bf (3)} For $i=2$, namely the $\varepsilon^{2\lambda}$ term, since $f^*$ and $f_1$ are independent of $Y$, the H-J-B equation \eqref{i-th HJB} with $i=2$ yields 
\begin{align*}
G_{2}(f^*,f_1,f_2)=-f^*_{\tau}+a_2f^*+a_1+qf_{2YY}=0,
\end{align*}
where $a_1$, $a_2$ are given in \eqref{def_a}. 
Since $f^*$ given by \eqref{mertoncrra f} satisfies 
\begin{align}\label{mertoncrra fe}
-\partial_{1}f^*+a_2f^*+a_1=0, 
\end{align}
we have $f_{2YY}=0$. Similar to the proof for the $\varepsilon^{\lambda}$ term, the boundary conditions \eqref{boundary condition ilambda=1} with $i=2$ leads to $2\lambda\neq 1$. Therefore, from \eqref{boundary condition ilambda neq 1} with $i=2$, $f_2$ is also independent of $Y$.

{\bf (4)} For $i=3$, namely the $\varepsilon^{3\lambda}$ term, based on the facts that $f^*$, $f_1$ and $f_2$ are independent of $Y$, the equation in no-trade region (see \eqref{i-th HJB} with $i=3$) can be expressed as 
\begin{align}\label{f3}
G_{3}(f^*,f_1,f_2,f_3)
=\,&\mathcal{L}f_1+qf_{3YY}+\gamma\sigma^2\Big(\frac{A}{1-\gamma}-\phi^*\Big)f^*\notag\\
=\,&\mathcal{L}f_1+qf_{3YY}=0,
\end{align}
where the operator $\mathcal{L}$ is defined as
\begin{align}
\mathcal{L}=\,&-\frac{\partial}{\partial{\tau}}+\left(r+A\sigma^2\phi^*+\gamma\sigma^2(\phi^*)^2-\frac{a_1}{f^{*}}\right)w\frac{\partial}{\partial w}+\frac{\sigma^2}{2}(\phi^*)^2w^2\frac{\partial^2}{\partial{w^2}}\notag\\
&+\Big(\frac{\gamma A}{1-\gamma}\sigma^2\phi^*-\frac{\gamma}{2}\sigma^2(\phi^*)^2+\frac{\gamma r-\beta}{1-\gamma}\Big)\notag\\
=\,&-\frac{\partial}{\partial{\tau}}+\left(r+\frac{A^2\sigma^2}{(1-\gamma)^2}-\frac{a_1}{f^{*}}\right)w\frac{\partial}{\partial w}+\frac{A^2\sigma^2}{2(1-\gamma)^2}w^2\frac{\partial^2}{\partial{w^2}}+a_2,\label{L operator}
\end{align}
and the constants $a_1$ and $a_2$ are defined in \eqref{def_a}.

Equation \eqref{f3} shows that $f_3$ is a quadratic function of $Y$: 
\begin{align}\label{f3 quadratic form}
qf_3(\tau,w,Y)=-\frac{1}{2}\mathcal{L}f_1(\tau,w)Y^2+c_1(\tau,w)Y+c_2(\tau,w)
\end{align}
with functions $c_1(\tau,w)$ and $c_2(\tau,w)$ to be determined.

We now examine boundary conditions for $f_3$. If $\lambda=\frac13$, from \eqref{boundary condition ilambda=1} with $i=3$, the boundary conditions for $f_3$ are
\begin{align}
&f_3(\tau,w,\widehat{Y}_B)-f_3(\tau,w,Y_B)=\frac{\gamma}{1-\gamma}K\left(\triangle Y_B,w\right)f^*(\tau),\label{V3b1st}\\
&\partial_{3}f_{3}(\tau,w,\widehat{Y}_B)=\partial_{3}f_{3}(\tau,w,Y_B)=\frac{\gamma}{1-\gamma}\,\partial_{1}K\left(\triangle Y_B,w\right)f^*(\tau),\label{V3b2nd}
\\
&f_3(\tau,w,\widehat{Y}_S)-f_3(\tau,w,Y_S)=\frac{\gamma}{1-\gamma}K\left(\triangle Y_S,w\right)f^*(\tau),\label{V3s1st}\\
&\partial_{3}f_{3}(\tau,w,\widehat{Y}_S)=\partial_{3}f_{3}(\tau,w,Y_S)=-\frac{\gamma}{1-\gamma}\,\partial_{1}K\left(\triangle Y_S,w\right)f^*(\tau).\label{V3s2nd}
\end{align}

Substituting \eqref{f3 quadratic form} into \eqref{V3b2nd}, we have
\begin{align*}
&-\mathcal{L}f_1(\tau,w)\widehat{Y}_B+c_1(\tau,w)=q\frac{\gamma}{1-\gamma}\,\partial_{1}K\left(\triangle Y_B,w\right)f^*(\tau),\\
&-\mathcal{L}f_1(\tau,w)Y_B+c_1(\tau,w)=q\frac{\gamma}{1-\gamma}\,\partial_{1}K\left(\triangle Y_B,w\right)f^*(\tau).
\end{align*}
Since $\widehat{Y}_B\neq Y_B$, the above two equations lead to 
\begin{align}\label{L f1}
\mathcal{L}f_1(\tau,w)=0. 
\end{align}
Equation~\eqref{L f1} can be derived from \eqref{f3 quadratic form} and \eqref{V3b2nd}, or from \eqref{f3 quadratic form} and \eqref{V3s2nd}. 

From \eqref{L f1}, \eqref{f3} is reduced to $f_{3YY}=0$, thus $f_3$ is a linear function of $Y$. 
Following the same procedure in our analysis for $i=1$ case, one will reach the conclusion that $\lambda\neq\frac{1}{3}$. Furthermore, \eqref{boundary condition ilambda neq 1} with $i=3$ implies that $f_3$ is independent of $Y$. 

The equation \eqref{L f1}, together with the boundary conditions \eqref{boundary condition ilambda neq 1} with $i=1$ and the initial condition $f_1(0,w)=0$, leads to $f_1=0$. Thus the leading-order correction to the value function occurs at least at ${\rm O}(\varepsilon^{2\lambda})$ level.

{\bf (5)} For $i=4$, namely the $\varepsilon^{4\lambda}$ term, from \eqref{i-th HJB} with $i=4$, the H-J-B equation for $f_4(\tau,w,Y)$ is 
\begin{align}\label{generalEV4Y}
G_4(f^*,f_1,f_2,f_3,f_4)
= qf_{4YY}+\mathcal{L}f_2(\tau,w)-\frac{1}{2}\gamma\sigma^2Y^2f^*(\tau)
=0,
\end{align} 
where $\mathcal{L}$ is given by \eqref{L operator}, $f_1=0$ and $f_3$ is independent of $Y$. 

After integrating \eqref{generalEV4Y} over $[0,Y]$, we obtain
\begin{align}\label{generalEV4once}
qf_{4Y}(\tau,w,Y)
=\frac{1}{6}\gamma\sigma^2f^*(\tau)Y^3-\mathcal{L}f_2(\tau,w)Y+N,
\end{align}
where $N\triangleq qf_{4Y}(\tau,w,0)$. 
Integrating \eqref{generalEV4once} over $[0,Y]$ again, we get
\begin{align}\label{linearEV4again}
qf_{4}(\tau,w,Y)
=\frac{1}{24}\gamma\sigma^2f^*(\tau)Y^4-\frac{1}{2}\mathcal{L}f_2(\tau,w)Y^2+NY+M,
\end{align}
where $M\triangleq qf_{4}(\tau,w,0)$.

Let us examine the possibility of $4\lambda=1$. By setting $4\lambda=1$, boundary conditions \eqref{boundary condition ilambda=1} with $i=4$ are 
\begin{align}
&f_4(\tau,w,\widehat{Y}_B)-f_4(\tau,w,Y_B)=\frac{\gamma}{1-\gamma}K\left(\triangle Y_B,w\right)f^*(\tau),\label{generalf4b1}\\
&\partial_{3}f_{4}(\tau,w,\widehat{Y}_B)=\partial_{3}f_{4}(\tau,w,Y_B)=\frac{\gamma}{1-\gamma}\,\partial_{1}K\left(\triangle Y_B,w\right)f^*(\tau),\label{generalf4b2}\\
&f_4(\tau,w,\widehat{Y}_S)-f_4(\tau,w,Y_S)=\frac{\gamma}{1-\gamma}K\left(\triangle Y_S,w\right)f^*(\tau),\label{generalf4s1}\\
&\partial_{3}f_{4}(\tau,w,\widehat{Y}_S)=\partial_{3}f_{4}(\tau,w,Y_S)=-\frac{\gamma}{1-\gamma}\,\partial_{1}K\left(\triangle Y_S,w\right)f^*(\tau).\label{generalf4s2}
\end{align}
The above four boundary conditions give us the following two results:
\begin{enumerate}[label=\roman*)]
	\item {\bf The leading term contribution from the transaction cost to the value function:} By setting $Y$ to $Y_B$ and to $\widehat{Y}_B$ in \eqref{generalEV4once}, we obtain the following two equations
	\begin{align*}
	q\,\partial_{3}f_{4}(\tau,w,Y_B)
	=\frac{1}{6}\gamma\sigma^2f^*(\tau)Y_B^3-\mathcal{L}f_2(\tau,w)Y_B+N,\\
	q\,\partial_{3}f_{4}(\tau,w,\widehat{Y}_B)
	=\frac{1}{6}\gamma\sigma^2f^*(\tau)\widehat{Y}_B^3-\mathcal{L}f_2(\tau,w)\widehat{Y}_B+N.
	\end{align*}
	These two equations and \eqref{generalf4b2} lead to
	\begin{align}	
	\mathcal{L}f_2(\tau,w)=\frac{1}{6}\gamma\sigma^2f^*(\tau)\left(Y_B^2+\widehat{Y}_B^2+Y_B\widehat{Y}_B\right).\label{generalf2B w}
	\end{align}
	Similarly, by setting $Y$ to $Y_S$ and to $\widehat{Y}_S$ in \eqref{generalEV4once}, \eqref{generalf4s2} leads to
	\begin{align}
	\mathcal{L}f_2(\tau,w)=\frac{1}{6}\gamma\sigma^2f^*(\tau)\left(Y_S^2+\widehat{Y}_S^2+Y_S\widehat{Y}_S\right).\label{generalf2S w}
	\end{align}	
	Equations \eqref{generalf2B w} and \eqref{generalf2S w} show that the leading-order buy and sell boundaries must be symmetric, namely
	\begin{align}\label{B=-S}
	Y_B=-Y_S\quad\text{and}\quad \widehat{Y}_B=-\widehat{Y}_S.
	\end{align}
	Equations \eqref{generalf2B w} and \eqref{generalf2S w} with the initial condition $f_2(0,w)=0$ and boundary conditions of $f_2$ (see \eqref{boundary condition ilambda neq 1} with $i=2$) imply that $f_2$ is independent of $w$. 
	Thus, \eqref{generalf2B w} and \eqref{generalf2S w} reduce to
	\begin{align}	
	-\partial_{1}f_{2}(\tau)+a_2f_2(\tau)=\frac{1}{6}\gamma\sigma^2f^*(\tau)\left(Y_B^2+\widehat{Y}_B^2+Y_B\widehat{Y}_B\right),\label{generalf2B}\\
	-\partial_{1}f_{2}(\tau)+a_2f_2(\tau)=\frac{1}{6}\gamma\sigma^2f^*(\tau)\left(Y_S^2+\widehat{Y}_S^2+Y_S\widehat{Y}_S\right).\label{generalf2S}
	\end{align}
	The solutions to \eqref{generalf2B} and \eqref{generalf2S} are
	\begin{align*}
	f_{2}(\tau)
	=\,&-4\gamma\sigma^2\psi(\tau)\left(Y_B^2+\widehat{Y}_B^2+Y_B\widehat{Y}_B\right),
	\\
	f_{2}(\tau)
	=\,&-4\gamma\sigma^2\psi(\tau)\left(Y_S^2+\widehat{Y}_S^2+Y_S\widehat{Y}_S\right),
	\end{align*}
	where \eqref{mertoncrra f} is used, and 
	\begin{align*}
	\psi(\tau)
	=\,&\frac{1}{24}\left[e^{a_2\tau}\frac{a_2(a_1+a_2a_3)\tau-a_1}{a_2^2}+\frac{a_1}{a_2^2}\right],
	\end{align*}
	that is \eqref{hatI}. The constants $a_1$, $a_2$, $a_3$, $a_4$ are defined in \eqref{def_a}. 
	
	By \eqref{generalYeta boundary}, it is worth noting that
	\begin{align}
	\varepsilon^{2\lambda}f_{2}(\tau)
	=\,&-4\gamma\sigma^2\psi(\tau)\left[\left(\varepsilon^{\lambda}Y_B\right)^2+\left(\varepsilon^{\lambda}\widehat{Y}_B\right)^2+\left(\varepsilon^{\lambda}Y_B\right)\left(\varepsilon^{\lambda}\widehat{Y}_B\right)\right]\notag\\
	=\,&-4\gamma\sigma^2\psi(\tau)\left[\left(\phi_B-\phi^*\right)^2+\left(\widehat{\phi}_B-\phi^*\right)^2+\left(\phi_B-\phi^*\right)\left(\widehat{\phi}_B-\phi^*\right)\right],\label{generalf2b}\\
	\varepsilon^{2\lambda}f_{2}(\tau)
	=\,&-4\gamma\sigma^2\psi(\tau)\left[\left(\varepsilon^{\lambda}Y_S\right)^2+\left(\varepsilon^{\lambda}\widehat{Y}_S\right)^2+\left(\varepsilon^{\lambda}Y_S\right)\left(\varepsilon^{\lambda}\widehat{Y}_S\right)\right]\notag\\
	=\,&-4\gamma\sigma^2\psi(\tau)\left[\left(\phi_S-\phi^*\right)^2+\left(\widehat{\phi}_S-\phi^*\right)^2+\left(\phi_S-\phi^*\right)\left(\widehat{\phi}_S-\phi^*\right)\right].\label{generalf2s}
	\end{align}

	\item {\bf The optimal trading-boundary:} Based on \eqref{generalEV4once} and \eqref{generalf2B w}, we eliminate $N$ and $\mathcal{L}f_2(\tau,w)$ in \eqref{linearEV4again}. Thus, \eqref{linearEV4again} can be rewritten as
	\begin{align}
	qf_4(\tau,w,Y)=&\,-\frac{1}{8}\gamma\sigma^2f^*(\tau)Y^4+\frac{1}{12}\gamma\sigma^2f^*(\tau)\left(Y_B^2+\widehat{Y}_B^2+Y_B\widehat{Y}_B\right)Y^2\notag\\
	&+qY\,\partial_{3}f_{4}(\tau,w,Y)+M.\label{generalEV4again b}
	\end{align}
	Using \eqref{q}, \eqref{generalf4b1}, \eqref{generalf4b2} and letting $\triangle Y_B=\widehat{Y}_B-Y_B$, \eqref{generalEV4again b} becomes
	\begin{align}\label{generalB1}
	(\phi^{*})^2\left(1-\frac{A}{1-\gamma}\right)^2\big[K\left(\triangle Y_B,w\right)-\triangle Y_B\,\partial_{1}K(\triangle Y_B,w)\big]
	=\,-\frac{1-\gamma}{12}\left(Y_B+\widehat{Y}_B\right)\left(\widehat{Y}_B-Y_B\right)^3.
	\end{align}
	By \eqref{generalEV4once}, \eqref{generalf4b2} and \eqref{generalf2B w}, we obtain $N=0$, and
	\begin{align}\label{generalB2}
	(\phi^{*})^2\left(1-\frac{A}{1-\gamma}\right)^2\,\partial_{1}K(\triangle Y_B,w)=-\frac{1-\gamma}{3}Y_B\widehat{Y}_B\left(Y_B+\widehat{Y}_B\right).
	\end{align}
	Similarly, following the same procedure for buy boundaries, we get the following equations for sell boundaries,
	\begin{align}
	H\big[K\left(\triangle Y_S,w\right)-\triangle Y_S\,\partial_{1}K(\triangle Y_S,w)\big]=&\,-\left(Y_S+\widehat{Y}_S\right)\left(\widehat{Y}_S-Y_S\right)^3,\label{generalS1}\\
	H\,\partial_{1}K(\triangle Y_S,w)=&\,4Y_S\widehat{Y}_S\left(Y_S+\widehat{Y}_S\right),\label{generalS2}
	\end{align}
	where $\triangle Y_S=Y_S-\widehat{Y}_S$ and $H$ is defined in \eqref{general A1 A2}. The equations \eqref{generalB1}, \eqref{generalB2}, \eqref{generalS1} and \eqref{generalS2} determine the optimal trading-boundary. Now, we analyze them in detail. 
	
	Based on \eqref{B=-S}, we introduce
	\begin{equation}\label{def x+-}
	\left\{\begin{aligned}
	&x^{-}=\,Y_S-\widehat{Y}_S=-Y_B+\widehat{Y}_B,\\
	&x^{+}=\,Y_S+\widehat{Y}_S=-(Y_B+\widehat{Y}_B).
	\end{aligned}\right.
	\end{equation}
	The definitions imply that $x^{-}\geq 0$, $x^{+}>0$. We define 
	\begin{align}\label{A12}
	\bar{A}_1(x^{-},w)=HK(x^{-},w)-x^{-}H\,\partial_{1}K(x^{-},w)\ \ \text{and}\ \  
	\bar{A}_2(x^{-},w)=H\,\partial_{1}K(x^{-},w),
	\end{align}
	where $H$ is defined in \eqref{general A1 A2}.
	
	Equations \eqref{generalB1} and \eqref{generalS1} lead to 
	\begin{align}\label{original x+- 1}
		\bar{A}_1(x^{-},w)=x^{+}(x^-)^3.
	\end{align}
	Similarly, \eqref{generalB2} and \eqref{generalS2} lead to 
	\begin{align}\label{original x+- 2}
		\bar{A}_2(x^{-},w)=\left((x^+)^2-(x^-)^2\right)x^{+}.
	\end{align}

	We introduce the variables
	\begin{align}\label{def x12}
	x=\varepsilon^{\lambda}x^{-},\ \ \widetilde{x}=\varepsilon^{\lambda}x^{+}.
	\end{align}
	Then, from \eqref{def x+-} and \eqref{def x12}, we have
	\begin{align}\label{varepsilonY}
	\varepsilon^{\lambda}Y_B=-\frac12(\widetilde{x}+x),\ \ \varepsilon^{\lambda}\widehat{Y}_B=-\frac12(\widetilde{x}-x),\ \ 
	\varepsilon^{\lambda}Y_S=\frac12(\widetilde{x}+x),\ \ \varepsilon^{\lambda}\widehat{Y}_S=\frac12(\widetilde{x}-x).
	\end{align}
	Furthermore, from \eqref{general cost}, we have 
	\begin{align}\label{xmK_xm}
	\varepsilon^{4\lambda} x^{-}\,\partial_{1}K=\varepsilon^{4\lambda-\lambda}x\cdot \varepsilon^{\lambda}\,\partial_{1}\bar{k}
	=\varepsilon^{4\lambda-1}x\,k',
	\end{align}
	and 
	\begin{align}\label{K_xm}
	\varepsilon^{3\lambda}\,\partial_{1}K =\varepsilon^{3\lambda}\cdot\varepsilon^{\lambda}\,\partial_{1}\bar{k}
	=\varepsilon^{4\lambda-1}k'.
	\end{align}
	Here we use $\partial_{1}$ and $k'$ to denote the first derivatives of the associated function and the transaction cost function, respectively. 
	
	Since $4\lambda=1$, \eqref{xmK_xm} and \eqref{K_xm} become 
	\begin{align}\label{xk_x k_x}
	\varepsilon^{4\lambda} x^{-}\,\partial_{1}K = x\,k',\ \ \text{and}\ \ \varepsilon^{3\lambda}\,\partial_{1}K = k'.
	\end{align}
	Finally, based on \eqref{def x12}, \eqref{varepsilonY} and \eqref{xk_x k_x}, \eqref{original x+- 1}-\eqref{original x+- 2} lead to \eqref{original generalx}. This confirms $\lambda = \frac14$. A combination of \eqref{generalf2b}, \eqref{generalf2s} and \eqref{varepsilonY} leads to the identity \begin{align*}
		\widehat{f}(\tau,w)
		=\,&\varepsilon^{2\lambda}f_2(\tau)
		=\,-\gamma\sigma^2 \psi(\tau) \left(3\widetilde{x}^2+x^2\right),
	\end{align*}
	that is \eqref{crraV2}. Similarly, a combination of \eqref{varepsilonY} and \eqref{generalYeta boundary} leads to \eqref{optimal sell boundary}-\eqref{optimal buy boundary}. 	
\end{enumerate}	

From \eqref{crraV} and \eqref{generalfYexpasion}, we have
\begin{align}
V=\,&e^{-\beta (T-\tau)}\frac{a_w w^{\gamma}}{\gamma}\left[f^*(\tau)+\varepsilon^{2\lambda}f_2(\tau)+\cdots\right]^{1-\gamma}\notag\\
=\,&V^*(\tau,w)\left[1+(1-\gamma)\frac{\varepsilon^{2\lambda}f_2(\tau)}{f^{*}(\tau)}+\cdots\right]\notag\\
=\,&V^*(\tau,w)\left[1+(1-\gamma)\frac{\widehat{f}(\tau,w)}{f^{*}(\tau)}+\cdots\right].\label{generalresultV}
\end{align}
This is \eqref{crraleadingV}. 

Expression \eqref{crraoptimalc} contains $\frac{1}{f}$, based on \eqref{generalfYexpasion}, it has an expansion
\begin{align}\label{f-}
\frac{1}{f}
=\,&\frac{1}{f^*}-\varepsilon^{\lambda}\frac{f_1}{(f^*)^{2}}+\varepsilon^{2\lambda}\left[-\frac{f_2}{(f^*)^{2}}+\frac{(f_1)^2}{(f^*)^{3}}\right]
	+\varepsilon^{3\lambda}\left[-\frac{f_3}{(f^*)^{2}}+2\frac{f_1f_2}{(f^*)^{3}}-\frac{(f_1)^3}{(f^*)^{4}}\right]\notag\\
&+\varepsilon^{4\lambda}\left[-\frac{f_4}{(f^*)^{2}}+2\frac{f_1f_3}{(f^*)^{3}}+\frac{(f_2)^2}{(f^*)^{3}}-3\frac{(f_1)^2f_2}{(f^*)^{4}}\right]+\cdots.
\end{align}
From \eqref{crraoptimalc}, \eqref{f-} and \eqref{generalfYexpasion}, we have
\begin{align}
c=\,&\frac{a_1w}{f}\left[1+\frac{1-\gamma}{\gamma}\frac{w\,\partial_{2}f(\tau,w,Y)-Y\,\partial_{3}f(\tau,w,Y)-\varepsilon^{-\lambda}\phi^*\,\partial_{3}f(\tau,w,Y)}{f}\right]^{\frac{1}{\gamma-1}}\notag\\
=\,&\frac{a_1w}{f^*}-\frac{a_1w}{(f^*)^{2}}\left[\widehat{f}(\tau,w)+\frac{1}{\gamma}w\,\partial_{2}\widehat{f}(\tau,w)\right]+\cdots\notag\\
=\,&c^*(\tau,w)-\frac{a_1w}{(f^*)^{2}}\left[\widehat{f}(\tau,w)+\frac{1}{\gamma}w\,\partial_{2}\widehat{f}(\tau,w)\right]+\cdots.\label{generalresultc}
\end{align}
This is \eqref{crraleadingc}. 

This completes our proof for Theorem \ref{General}.

\section{Perturbation expansion for the governing equations and boundary conditions}\label{appendix: details in no-trade}
In this appendix, we show the details of the derivation for the perturbation expansion of the governing equations in the no-trade region and the associated boundary conditions for the case of general transaction cost structure $k(\cdot)$.

\subsection{Derivation for the governing equation}\label{appendix: chang variables equation}
After changing variables from $(\tau,w,\phi)$ to $(\tau,w,Y)$ in \eqref{crraH-J-BVc*} and multiplying the result by $\varepsilon^{2\lambda}$, we have
\begin{align}\label{generalHJBVCY}
D_{0}+\varepsilon^{\lambda}D_{1}+\varepsilon^{2\lambda}(D_{2}+C)+\varepsilon^{3\lambda}D_{3}+\varepsilon^{4\lambda}D_{4}=0,
\end{align}
where 
\begin{align}
C=\,&a_1\left[1+\frac{1-\gamma}{\gamma}\frac{w\,\partial_{2}f-Y\,\partial_{3}f-\varepsilon^{-\lambda}\phi^*\,\partial_{3}f}{f}\right]^{\frac{\gamma}{\gamma-1}},\ f=f(\tau,w,Y),\label{C part}
\\
D_i=\,&D_i^l+D_i^{nl},\ \ i=0,1,2,3,4.\label{D0-4}
\end{align}
Here $D_i^l$ has a linear dependence on $f = f(\tau,w,Y)$ and $D_i^{nl}$ has a nonlinear dependence on $f = f(\tau,w,Y)$. Their explicit expressions are
\begin{align}
D_0^l 
=\,& \frac{1}{2}\sigma^2\phi^{*^2}(1-\phi^{*})^2\,\partial_{33}f,\label{D0 l}
\\
D_{0}^{nl}=\,&-\frac{1}{2}\sigma^2\phi^{*^2}(1-\phi^{*})^2\gamma\frac{(\partial_{3}f)^2}{f},\label{D0 nl}
\\
D_1^l  
=\,& \sigma^2\phi^*\left(1-\phi^*\right)[A+(\gamma-1)\phi^*]\,\partial_{3}f+\sigma^2\phi^*(1-\phi^*)(1-2\phi^*)Y\,\partial_{33}f\notag\\
&+\sigma^2\phi^{*^2}(1-\phi^{*})w\,\partial_{23}f,\label{D1 l}
\\
D_{1}^{nl}=\,&-\phi^* \frac{\sigma^2\gamma}{f}\left[(1-\phi^*)(1-2\phi^*)Y(\partial_{3}f)^2+\phi^{*}(1-\phi^{*})w\,\partial_{3}f\partial_{2}f\right],\label{D1 nl}
\\
D_2^l 
=\,& -\partial_{1}f+\left(\frac{\gamma r-\beta}{1-\gamma}+\sigma^2\frac{\gamma A}{1-\gamma}\phi^{*}-\frac{\gamma}{2}\sigma^2\phi^{*^2}\right)f+\left(r+A\sigma^2\phi^{*}+\gamma\sigma^2\phi^{*^2}\right)w\,\partial_{2}f\notag\\
&+\sigma^2\left[(1-2\phi^{*})A+(\gamma-1)(2-3\phi^*)\phi^*\right]Y\,\partial_{3}f
+\frac{1}{2}\sigma^2\left[(1-2\phi^{*})^2-2\phi^{*}(1-\phi^{*})\right]Y^2\,\partial_{33}f\notag\\
&+\frac{1}{2}\sigma^2\phi^{*^2}w^2\,\partial_{22}f+\sigma^2\phi^{*}(2-3\phi^*)wY\,\partial_{23}f,\label{D2 l}
\\
D_2^{nl}=\,&- \frac{\sigma^2\gamma}{2f}\left\{\left[(1-2\phi^{*})^2-2\phi^{*}(1-\phi^{*})\right]Y^2(\partial_{3}f)^2+\phi^{*^2}w^2(\partial_{2}f)^2+2\phi^{*}(2-3\phi^*)wY\,\partial_{2}f\partial_{3}f\right\},\label{D2 nl}
\\
D_3^l 
=\,& -\sigma^2\left[(1-3\phi^*)(1-\gamma)+A\right]Y^2\,\partial_{3}f+\sigma^2(2\phi^{*}-1)Y^3\,\partial_{33}f
+\sigma^2\phi^{*}Yw^2\,\partial_{22}f\notag\\
&+\sigma^2(1-3\phi^{*})wY^2\,\partial_{23}f+\sigma^2(A+2\gamma\phi^*)Yw\,\partial_{2}f+\sigma^2\gamma\left(\frac{A}{1-\gamma}-\phi^{*}\right)Yf,\label{D3 l}
\\
D_3^{nl}=\,&-\sigma^2(2\phi^{*}-1)Y^3\gamma\frac{(\partial_{3}f)^2}{f}-\sigma^2\phi^{*}Yw^2\gamma \frac{(\partial_{2}f)^2}{f}-\sigma^2(1-3\phi^{*})wY^2\gamma\frac{\,\partial_{2}f\partial_{3}f}{f},\label{D3 nl}
\\
D_4^l  
=\,& -\frac{1}{2}\gamma\sigma^2Y^2f
+(1-\gamma)\sigma^2Y^3\,\partial_{3}f+\frac{1}{2}\sigma^2Y^4\,\partial_{33}f+\frac{1}{2}\sigma^2Y^2w^2\,\partial_{22}f
-\sigma^2Y^3w\,\partial_{23}f+\gamma\sigma^2Y^2w\,\partial_{2}f,\label{D4 l}
\\
D_4^{nl}=\,&-\frac{1}{2}\sigma^2Y^4\gamma \frac{(\,\partial_{3}f)^2}{f}-\frac{1}{2}\sigma^2Y^2w^2\gamma \frac{(\partial_{2}f)^2}{f}+\sigma^2Y^3w\gamma \frac{\,\partial_{2}f\partial_{3}f}{f},\label{D4 nl}
\end{align} 
where $a_1$ is defined by \eqref{def_a}.

As shown by \eqref{C part}-\eqref{D0-4}, $C$, $D_i$, $i=0,1,2,3,4$, not only depend on $\varepsilon$ explicitly, but also implicitly through the functions $f$ and $f^{-1}$. Therefore, we need to further expand these dependencies. 

In the following, we simplify the notation as $f_i = f_i(\tau,w,Y)$, $i=1,2,\dots$, and $f^* = f^*(\tau)$. 
After substituting $f$ given by \eqref{generalfYexpasion} and $f^{-1}$ given by \eqref{f-} into \eqref{C part}, we obtain the following expansion for $C$
\begin{align*}
&C\left(f^*+\varepsilon^{\lambda}f_1+\varepsilon^{2\lambda}f_2+\varepsilon^{3\lambda}f_3+\varepsilon^{4\lambda} f_4+\cdots\right)\\
=\,&a_1\Bigg\{1+\frac{1-\gamma}{\gamma}\left[\frac{1}{f^*}-\varepsilon^{\lambda}\frac{f_1}{(f^*)^2}+\varepsilon^{2\lambda}\frac{-f^*f_2+f_1^2}{(f^*)^3}+\cdots\right]\times\varPi(f^*,f_1,f_2,\dots)\Bigg\}^{\frac{\gamma}{\gamma-1}},
\end{align*}
where $\varPi(f^*,f_1,f_2,\dots) = -\varepsilon^{-\lambda}\phi^*\,\partial_{3}f^*+(w\,\partial_{2}f^*-Y\,\partial_{3}f^*-\phi^*\,\partial_{3}f_{1})+\varepsilon^{\lambda}(w\,\partial_{2}f_{1}-Y\,\partial_{3}f_{1}-\phi^*\,\partial_{3}f_{2})+\varepsilon^{2\lambda}(w\,\partial_{2}f_{2}-Y\,\partial_{3}f_{2}-\phi^*\,\partial_{3}f_{3})+\cdots$. 
Since $f^*(\tau)$ does not depend on $w$ and $Y$, we have $\partial_{2}f^*=\partial_{3}f^*=0$, then the above equation can be expanded as
\begin{align}
&C\left(f^*+\varepsilon^{\lambda}f_1+\varepsilon^{2\lambda}f_2+\varepsilon^{3\lambda}f_3+\varepsilon^{4\lambda} f_4+\cdots\right)\notag\\
=\,&a_1\Bigg\{1+\frac{1-\gamma}{\gamma}\left[\frac{1}{f^*}-\varepsilon^{\lambda}\frac{f_1}{(f^*)^2}+\varepsilon^{2\lambda}\frac{-f^*f_2+f_1^2}{(f^*)^3}+\cdots\right]\times\Xi(f_1,f_2,f_3,\dots)\Bigg\}^{\frac{\gamma}{\gamma-1}}\notag\\
=\,&a_1\Bigg\{I_0(f^*,f_1)+\varepsilon^{\lambda}I_1(f^*,f_1,f_2)+\varepsilon^{2\lambda}I_2(f^*,f_1,f_2,f_3)+\cdots\Bigg\}^{\frac{\gamma}{\gamma-1}}\notag\\
=\,&C_0(f^*,f_1)+\varepsilon^{\lambda}C_1(f^*,f_1,f_2)+\varepsilon^{2\lambda}C_2(f^*,f_1,f_2,f_3)+\cdots,\label{C expansion}
\end{align}
where $\Xi(f_1,f_2,f_3,\dots) = (-\phi^*\,\partial_{3}f_{1})+\varepsilon^{\lambda}(w\,\partial_{2}f_{1}-Y\,\partial_{3}f_{1}-\phi^*\,\partial_{3}f_{2})+\varepsilon^{2\lambda}(w\,\partial_{2}f_{2}-Y\,\partial_{3}f_{2}-\phi^*\,\partial_{3}f_{3})+\cdots$, and
\begin{align*}
&I_0(f^*,f_1)=\,1-\frac{1-\gamma}{\gamma}\phi^*\frac{\partial_{3}f_{1}}{f^*},\\
&I_1(f^*,f_1,f_2)=\,\frac{1-\gamma}{\gamma}\left[-\phi^*\frac{\partial_{3}f_{2}}{f^*}+\frac{1}{f^*}\left(w\,\partial_{2}f_{1}+\left(\phi^*\frac{f_{1}}{f^*}-Y\right)\partial_{3}f_{1}\right)\right],
\\
&I_2(f^*,f_1,f_2,f_3)=\,\frac{1-\gamma}{\gamma}\left[-\phi^*\frac{\partial_{3}f_{3}}{f^*}+\frac{1}{f^*}\left(w\,\partial_{2}f_{2}+\left(\phi^*\frac{f_{1}}{f^*}-Y\right)\partial_{3}f_{2}+\phi^*\frac{f_{2}}{f^*}\partial_{3}f_{1}\right)\right],
\\
&C_0(f^*,f_1) =\, a_1I_0^{\frac{\gamma}{\gamma-1}}(f^*,f_1),\\
&C_1(f^*,f_1,f_2) =\, \frac{\gamma}{\gamma-1}a_1I_0^{\frac{1}{\gamma-1}}(f^*,f_1)I_1(f^*,f_1,f_2),\\
&C_2(f^*,f_1,f_2,f_3) =\, a_1\frac{\gamma}{\gamma-1}\left[I_0^{\frac{1}{\gamma-1}}(f^*,f_1)I_2(f^*,f_1,f_2,f_3)+\frac{1}{2(\gamma-1)}I_0^{\frac{2-\gamma}{\gamma-1}}(f^*,f_1)I_1^2(f^*,f_1,f_2)\right].
\end{align*}

After substituting \eqref{C expansion} and \eqref{generalfYexpasion} into \eqref{generalHJBVCY}, and regrouping them in terms of the power of $\varepsilon^{\lambda}$, \eqref{generalHJBVCY} becomes
\begin{align}
&D_{0}\left(f^*+\varepsilon^{\lambda}f_1+\varepsilon^{2\lambda}f_2+\varepsilon^{3\lambda}f_3+\varepsilon^{4\lambda} f_4+\cdots\right)\notag\\
&+\varepsilon^{\lambda}D_{1}\left(f^*+\varepsilon^{\lambda}f_1+\varepsilon^{2\lambda}f_2+\varepsilon^{3\lambda}f_3+\varepsilon^{4\lambda} f_4+\cdots\right)\notag\\
&+\varepsilon^{2\lambda}\left[D_{2}\left(f^*+\varepsilon^{\lambda}f_1+\varepsilon^{2\lambda}f_2+\varepsilon^{3\lambda}f_3+\varepsilon^{4\lambda} f_4+\cdots\right)+C_0(f^*,f_1)\right]\notag\\
&+\varepsilon^{3\lambda}\left[D_{3}\left(f^*+\varepsilon^{\lambda}f_1+\varepsilon^{2\lambda}f_2+\varepsilon^{3\lambda}f_3+\varepsilon^{4\lambda} f_4+\cdots\right)+C_1(f^*,f_1,f_2)\right]\notag\\
&+\varepsilon^{4\lambda}\left[D_{4}\left(f^*+\varepsilon^{\lambda}f_1+\varepsilon^{2\lambda}f_2+\varepsilon^{3\lambda}f_3+\varepsilon^{4\lambda} f_4+\cdots\right)+C_2(f^*,f_1,f_2,f_3)\right]+\cdots=0,\label{C D_i}
\end{align}
where $D_i$ for $i=0,1,2,3,4$ are given by \eqref{D0-4}.

By comparing \eqref{C D_i} with \eqref{G}, we have the expression for $G_j$ in \eqref{G}
\begin{equation}\label{G equals CD}
G_j=\,D_0^l(f_j)+\left\{\sum_{i=1}^{j\wedge4} D_i^l(f_{(j-i)\vee 0})+{\bf 1}_{\{j\geq2\}}\left[\widehat{D}_{j-2}(f_0,\dots,f_{j-1})+C_{j-2}(f_0,\dots,f_{j-1})\right]\right\}.
\end{equation}
Here $j=0,1,2,3,\dots$, $j\wedge4\triangleq \min\{j,4\}$, $(j-i)\vee 0\triangleq \max\{j-i,0\}$, $f_0=f^*$, and $\widehat{D}_{j-2}(f_0,\dots,f_{j-1})$ for $j=2,3,4$ are given by
\begin{align*}
\widehat{D}_0(f^*,f_1)
=\,&-\frac{1}{2}\sigma^2\phi^{*^2}(1-\phi^{*})^2\gamma\frac{(\partial_{3}f_{1})^2}{f^*},
\\
\widehat{D}_1(f^*,f_1,f_2)
=\,&-\frac{1}{2}\sigma^2\phi^{*^2}(1-\phi^{*})^2\gamma\left[2\,\partial_{3}f_{2}-f_1\frac{\partial_{3}f_{1}}{f^*}\right]\frac{\partial_{3}f_{1}}{f^*}\\
&-\sigma^2\phi^*(1-\phi^*)\gamma \Big[(1-2\phi^*)Y\,\partial_{3}f_{1}-\phi^{*}w\,\partial_{2}f_{1}\Big]\frac{\partial_{3}f_{1}}{f^*},
\end{align*}
and
\begin{align*}
&\widehat{D}_2(f^*,f_1,f_2,f_3)\\
=\,&-\sigma^2\phi^{*^2}(1-\phi^{*})^2\frac{\gamma}{2f^*}\left[(\partial_{3}f_{2})^2+2\,(\partial_{3}f_{1})(\partial_{3}f_{3})-2f_1\frac{\partial_{3}f_{1}}{f^*}\,\partial_{3}f_{2}+(f_1^2-f_2f^*)\left(\frac{\partial_{3}f_{1}}{f^*}\right)^2\right]\\
&-\sigma^2\phi^*(1-\phi^*)(1-2\phi^*)Y\frac{\gamma}{f^*}\left[2\,(\partial_{3}f_{1})(\partial_{3}f_{2})-f_1\frac{(\partial_{3}f_{1})^2}{f^*}\right]\\
&-\sigma^2\phi^{*^2}(1-\phi^{*})w\frac{\gamma}{f^*}\left[(\partial_{3}f_{1})(\partial_{2}f_{2})+(\partial_{2}f_{1})(\partial_{3}f_{2})-f_1\frac{\partial_{3}f_{1}}{f^*}\partial_{2}f_{1}\right]\\
&-\frac{1}{2}\sigma^2 \frac{\gamma}{f^*}\left\{\Big[(1-2\phi^{*})^2-2\phi^{*}(1-\phi^{*})\Big]Y^2(\partial_{3}f_{1})^2+\phi^{*^2}w^2(\partial_{2}f_{1})^2\right\}\\
&-\sigma^2\frac{\gamma}{f^*}\phi^{*}(2-3\phi^*)wY(\partial_{3}f_{1})(\partial_{2}f_{1}),
\end{align*}
which are the results of the nonlinear term $D_i^{nl}$  given by \eqref{D0 nl}, \eqref{D1 nl}, \eqref{D2 nl}, \eqref{D3 nl} and \eqref{D4 nl}. 

In \eqref{G equals CD}, $G_j$ is the governing equation for the order $\left(\varepsilon^{\lambda}\right)^j$ in the singular perturbation expansion and does not depend on $\varepsilon$. 
In \eqref{G equals CD}, only the first term $D_0^l(f_j)$ depends on the unknown function $f_j$, while all remaining terms in the bracket depend only on lower-order functions, which are known in the singular perturbation expansion procedure because we carry out the expansion from lower to higher orders. Thus, at each order we solve a linear equation.

\subsection{Derivation for the boundary conditions}\label{general B.C.} 
After changing variables from $(\tau,w,\phi)$ to $(\tau,w,Y)$, namely using the relationship \eqref{general cost} and \eqref{generalYeta boundary}, boundary conditions \eqref{crrabc1}-\eqref{crrasc3} become
\begin{align}
f(\tau,w,Y_B)=\,&\left(1+\varepsilon\triangle W_B\right)^{\frac{\gamma}{1-\gamma}}f(\tau,\widehat{w}_B,\widehat{Y}_B),\label{bc-1Y}
\\
\partial_{Y_B}f(\tau,w,Y_B)
=\,&\varepsilon\,\partial_{1}K\big(\triangle Y_B,w\big)\frac{\gamma}{1-\gamma}\left(1+\varepsilon\triangle W_B\right)^{-1}f(\tau,w,Y_B)\notag\\
&+\varepsilon\,\partial_{1}K\big(\triangle Y_B,w\big)\left(1+\varepsilon\triangle W_B\right)^{\frac{\gamma}{1-\gamma}}w\,\partial_{2}f(\tau,\widehat{w}_B,\widehat{Y}_B),\label{bc-2Y}
\\
\partial_{\widehat{Y}_B}f(\tau,\widehat{w}_B,\widehat{Y}_B)
=\,&\varepsilon\,\partial_{1}K\big(\triangle Y_B,w\big)w\,\partial_{2}f(\tau,\widehat{w}_B,\widehat{Y}_B)\notag\\
&+\varepsilon\,\partial_{1}K\big(\triangle Y_B,w\big)\frac{\gamma}{1-\gamma}\left(1+\varepsilon\triangle W_B\right)^{-\frac{1}{1-\gamma}}f(\tau,w,Y_B),\label{bc-3Y}
\\
f(\tau,w,Y_S)=\,&\left(1+\varepsilon\triangle W_S\right)^{\frac{\gamma}{1-\gamma}}f(\tau,\widehat{w}_S,\widehat{Y}_S),\label{sc-1Y}
\\
\partial_{Y_S}f(\tau,w,Y_S)
=\,&-\varepsilon\,\partial_{1}K\big(\triangle Y_S,w\big)\frac{\gamma}{1-\gamma}\left(1+\varepsilon\triangle W_S\right)^{-1}f(\tau,w,Y_S)\notag\\
&-\varepsilon\,\partial_{1}K\big(\triangle Y_S,w\big)\left(1+\varepsilon\triangle W_S\right)^{\frac{\gamma}{1-\gamma}}w\,\partial_{2}f(\tau,\widehat{w}_S,\widehat{Y}_S),\label{sc-2Y}
\\
\partial_{\widehat{Y}_S}f(\tau,\widehat{w}_S,\widehat{Y}_S)=\,& -\varepsilon\,\partial_{1}K\big(\triangle Y_S,w\big)w\,\partial_{2}f(\tau,\widehat{w}_S,\widehat{Y}_S)\notag\\
&-\varepsilon\,\partial_{1}K\big(\triangle Y_S,w\big)\frac{\gamma}{1-\gamma}\left(1+\varepsilon\triangle W_S\right)^{-\frac{1}{1-\gamma}}f(\tau,w,Y_S),\label{sc-3Y}
\end{align}
where $\widehat{w}_B=w(1+\varepsilon\triangle W_B)$, $\triangle W_B=-K\big(\triangle Y_B,w\big)$, $\triangle Y_B=\widehat{Y}_B-Y_B$, $\widehat{w}_S=w(1+\varepsilon\triangle W_S)$, $\triangle W_S=-K\big(\triangle Y_S,w\big)$, and $\triangle Y_S=Y_S-\widehat{Y}_S$. 

Since $\varepsilon$ is small, we apply Taylor expansion in terms of $\varepsilon$ to \eqref{bc-1Y}-\eqref{sc-3Y} and only keep the leading order terms. The results are \eqref{crraBC1}-\eqref{crraSC3}.

In these expansions, $f$, $\widehat{w}_B$ and $\widehat{w}_S$ in \eqref{bc-1Y}-\eqref{sc-3Y} all depend on $\varepsilon$. We have expanded these quantities as well in the above derivation.

This completes our derivations for the singular perturbation expansion of the governing equation and boundary conditions for the case of the general transaction cost structure.

\section{Proof of recovery of second order boundary conditions in the case of proportional cost only}\label{appendix: proof fixed or prop}

In the case of proportional cost only ($A_1=0$) given by Corollary \ref{CRRA proportional}, we need to take extra care to analyze the associated \textit{boundary conditions} since in this case the post-buy boundary coincides with the pre-buy boundary and the post-sell boundary coincides with the pre-sell boundary, namely $\widehat{Y}_B=Y_B$ and $\widehat{Y}_S=Y_S$ (see Remark \ref{Remark special cases of original generalx} in Section \ref{sec: general costs} and \eqref{generalYeta boundary}). This means the investor trades an infinitesimally small amount of wealth $\delta_B$ or $\delta_S$ as soon as the portfolio position lies outside the boundaries of the no-trade region. Thus, the change of wealth due to the transaction costs is $\triangle W_B=-K_2\delta_B$ or $\triangle W_S=-K_2\delta_S$ in \eqref{crraBC1}-\eqref{crraSC3}. 
Equations \eqref{crraBC1} and \eqref{crraSC1} give
\begin{align*}
	f(\tau,w,Y_B,\varepsilon^{\lambda})=\,&f(\tau,w,Y_B+\delta_B,\varepsilon^{\lambda})-\varepsilon K_2\delta_B\left[\frac{\gamma}{1-\gamma}f(\tau,w,Y_B+\delta_B,\varepsilon^{\lambda})+w\,\partial_{2}f(\tau,w,Y_B+\delta_B,\varepsilon^{\lambda})\right],\\
	f(\tau,w,Y_S,\varepsilon^{\lambda})=\,&f(\tau,w,Y_S-\delta_S,\varepsilon^{\lambda})-\varepsilon K_2\delta_S\left[\frac{\gamma}{1-\gamma}f(\tau,w,Y_S-\delta_S,\varepsilon^{\lambda})+w\,\partial_{2}f(\tau,w,Y_S-\delta_S,\varepsilon^{\lambda})\right].
\end{align*}
After applying the variation principle and taking the limits $\delta_B\rightarrow 0$ and $\delta_S\rightarrow 0$ (because the trading volume is infinitesimal), the above two equations become
\begin{align}
	\partial_{Y_B}f(\tau,w,Y_B,\varepsilon^{\lambda})=\,&\varepsilon K_2\left[\frac{\gamma}{1-\gamma}f(\tau,w,Y_B,\varepsilon^{\lambda})+w\,\partial_{2}f(\tau,w,Y_B,\varepsilon^{\lambda})+{\rm O}(\delta_B)\right],\label{crraBC1-p}\\
	\partial_{Y_S}f(\tau,w,Y_S,\varepsilon^{\lambda})=\,&-\varepsilon K_2\left[\frac{\gamma}{1-\gamma}f(\tau,w,Y_S,\varepsilon^{\lambda})+w\,\partial_{2}f(\tau,w,Y_S,\varepsilon^{\lambda})+{\rm O}(\delta_S)\right],\label{crraSC1-p}
\end{align}
which are the special case of \eqref{crraBC2} and \eqref{crraSC2}. Following a similar procedure, \eqref{crraBC3} and \eqref{crraSC3} lead to the boundary conditions
\begin{align}		
	\partial_{Y_BY_B}f(\tau,w,Y_B,\varepsilon^{\lambda})=\,&2\varepsilon K_2w\left[\partial_{23}f(\tau,w,Y_B,\varepsilon^{\lambda})+{\rm O}(\delta_B)\right],\label{crraBC2-p}\\ 
	\partial_{Y_SY_S}f(\tau,w,Y_S,\varepsilon^{\lambda})=\,&-2\varepsilon K_2w\left[\partial_{23}f(\tau,w,Y_S,\varepsilon^{\lambda})+{\rm O}(\delta_S)\right],\label{crraSC2-p}
\end{align}
where we have applied \eqref{crraBC1-p} and \eqref{crraSC1-p}. 

In summary, boundary conditions for the case of only proportional cost are \eqref{crraBC1-p}-\eqref{crraSC2-p}.

\section{Proof for Theorem \ref{effect of costs}}\label{appendix: proof effect of costs}
In this section, we prove Theorem \ref{effect of costs}. 

{\bf Proof for part \ref{effect of A1}.}	
For fixed $A_2$, \eqref{original generalx} gives
\begin{align}
&2x\widetilde{x}{\rm d}x\ =\left(3\widetilde{x}^{2}-x^{2}\right){\rm d}\widetilde{x},\label{derivation 5}\\
&3x^2\widetilde{x}{\rm d}x+x^3{\rm d}\widetilde{x}\ ={\rm d}A_1.\label{derivation 6}
\end{align}

Substituting \eqref{derivation 5} into \eqref{derivation 6}, we get
\begin{align}\label{derivation 8}
\frac{{\rm d}\widetilde{x}}{{\rm d}A_1} 
=\, \frac{2}{x}\frac{1}{9\widetilde{x}^{2}-x^{2}}
=\, \frac{2}{x}\frac{1}{8\widetilde{x}^{2}+\left(\widetilde{x}^{2}-x^{2}\right)}.
\end{align}
From the first equation of \eqref{original generalx}, \eqref{derivation 8} can be written as
\begin{align*}
\frac{{\rm d}\widetilde{x}}{{\rm d}A_1} = \frac{2\widetilde{x}}{x\left(8\widetilde{x}^{3}+A_2\right)}>0,
\end{align*} 
which implies
$\frac{{\rm d}\widetilde{x}}{{\rm d}A_1}>0$.  
From \eqref{derivation 5} and the first equation of \eqref{original generalx}, we have
\begin{align*}
\frac{{\rm d}x}{{\rm d}A_1}=\frac{{\rm d}x}{{\rm d}\widetilde{x}}\frac{{\rm d}\widetilde{x}}{{\rm d}A_1}
=\frac{{\rm d}\widetilde{x}}{{\rm d}A_1}\frac{3\widetilde{x}^{2}-x^{2}}{2x\widetilde{x}}
=\frac{{\rm d}\widetilde{x}}{{\rm d}A_1}\frac{2\widetilde{x}^{3}+A_2}{2x\widetilde{x}^2}
>0.
\end{align*}
In the above expression, the second equality comes from \eqref{derivation 5} and the last equality comes from the first equation of \eqref{original generalx}. It follows that $\frac{{\rm d}(\widetilde{x}+x)}{{\rm d}A_1}>0$. Therefore, from \eqref{optimal sell boundary}-\eqref{optimal buy boundary}, $\phi_S$ increases with $A_1$, and $\phi_B$ decreases with $A_1$.

From \eqref{derivation 5}, we also have
\begin{align}\label{dx-x dA1}
2x\widetilde{x}\frac{{\rm d}(\widetilde{x}-x)}{{\rm d}A_1}=\left(2x\widetilde{x}-(3\widetilde{x}^{2}-x^{2})\right)\frac{{\rm d}\widetilde{x}}{{\rm d}A_1}.
\end{align}
From the first equation of \eqref{original generalx}, \eqref{dx-x dA1} can be written as 
\begin{align*}
\frac{{\rm d}(\widetilde{x}-x)}{{\rm d}A_1} = -\frac{\widetilde{x}(\widetilde{x}-x)^2+2A_2}{2x\widetilde{x}^2}\frac{{\rm d}\widetilde{x}}{{\rm d}A_1}<0.
\end{align*}
Therefore, $\frac{{\rm d}(\widetilde{x}-x)}{{\rm d}A_1}<0$. This means that $\widehat{\phi}_S$ decreases with $A_1$, and $\widehat{\phi}_B$ increases with $A_1$.

From $\frac{{\rm d}x}{{\rm d}A_1}>0$, \eqref{optimal sell boundary}-\eqref{optimal buy boundary} show that both the buy trading sizes $|\widehat{\phi}_B-\phi_B|$ and the sell trading size $|\widehat{\phi}_S-\phi_S|$ increase with $A_1$.

In summary, $\phi_S$ and $\widehat{\phi}_B$ increase with $A_1$, $\phi_B$ and $\widehat{\phi}_S$ decrease with $A_1$, and both buy trading size $|\widehat{\phi}_B-\phi_B|$ and sell trading size $|\widehat{\phi}_S-\phi_S|$ increase with $A_1$. This completes our proof for part \ref{effect of A1}.

{\bf Proof for part \ref{effect of A2}.} 
Now, we consider how these quantities vary with $A_2$ when $A_1$ is fixed. Equation \eqref{original generalx} gives
\begin{align}
&\left(3\widetilde{x}^{2}-x^{2}\right){\rm d}\widetilde{x}-2x\widetilde{x} {\rm d}x={\rm d}A_2,\label{derivation 1}\\
&\widetilde{x} {\rm d}x=-\frac13x {\rm d}\widetilde{x}.\label{derivation 2}
\end{align}
Substituting \eqref{derivation 2} into \eqref{derivation 1}, we get
\begin{align}
{\rm d}A_2 
= \left(3\widetilde{x}^{2}-\frac13x^{2}\right){\rm d}\widetilde{x}
=\left[\frac83\widetilde{x}^{2}+\frac13\left(\widetilde{x}^{2}-x^{2}\right)\right]{\rm d}\widetilde{x}
= \left[\frac83\widetilde{x}^{2}+\frac{A_2}{3\widetilde{x}}\right]{\rm d}\widetilde{x}.\label{derivation 3}
\end{align}
The last equality follows from the first equation in \eqref{original generalx}. Equation \eqref{derivation 3} implies that $\frac{{{\rm d}\widetilde{x}}}{{\rm d}A_2}>0$. From \eqref{derivation 2}, we have $\frac{{{\rm d}x}}{{\rm d}A_2}=\frac{{{\rm d}x}}{{{\rm d}\widetilde{x}}}\frac{{{\rm d}\widetilde{x}}}{{\rm d}A_2}=-\frac{x}{3\widetilde{x}}\frac{{{\rm d}\widetilde{x}}}{{\rm d}A_2}<0$. Here the fact $\widetilde{x}\neq 0$ is used. By \eqref{optimal sell boundary}-\eqref{optimal buy boundary}, the trading size for sell and buy are
\begin{align*}
|\widehat{\phi}_S-\phi_S|=|\widehat{\phi}_B-\phi_B|=x.
\end{align*}
Therefore, the trading sizes decrease with $A_2$, for a given $A_1$.

We now answer how $\phi_B$, $\widehat{\phi}_B$, $\phi_S$, and $\widehat{\phi}_S$ will vary with $A_2$ for a given $A_1$.

Based on $\frac{{{\rm d}\widetilde{x}}}{{\rm d}A_2}>0$ and $\frac{{{\rm d}x}}{{\rm d}A_2}<0$, we have $\frac{{{\rm d}(\widetilde{x}-x)}}{{\rm d}A_2}>0$. From \eqref{optimal sell boundary}-\eqref{optimal buy boundary}, it follows that $\widehat{\phi}_S$ increases with $A_2$, and $\widehat{\phi}_B$ decreases with $A_2$. From \eqref{derivation 2} and \eqref{derivation 3}, it is clear that
\begin{align}\label{dx+x dA2}
\frac{{\rm d}(\widetilde{x}+x)}{{\rm d}A_2}
=\frac{3\widetilde{x}-x}{8\widetilde{x}^{3}+A_2}.
\end{align}
Combining \eqref{original generalx} with the fact that $x$ and $\widetilde{x}$ are nonnegative, we have
\begin{align}\label{dx+x dA2-1}
3\widetilde{x}-x 
=2\widetilde{x}+\frac{\widetilde{x}^2-x^2}{\widetilde{x}+x}\cdot\frac{\widetilde{x}}{\widetilde{x}}
=2\widetilde{x}+\frac{A_2}{\widetilde{x}+x}\cdot\frac{1}{\widetilde{x}}>0,
\end{align}
After substituting \eqref{dx+x dA2-1} into \eqref{dx+x dA2}, we obtain $\frac{{\rm d}(\widetilde{x}+x)}{{\rm d}A_2}>0$. Thus, $\phi_S$ increases with $A_2$ and $\phi_B$ decreases with $A_2$.

In summary, $\phi_S$ and $\widehat{\phi}_S$ increase with $A_2$, $\phi_B$ and $\widehat{\phi}_B$ decrease with $A_2$, and both buy trading size $|\widehat{\phi}_B-\phi_B|$ and sell trading size $|\widehat{\phi}_S-\phi_S|$ decrease with $A_2$.

This completes our proof for Theorem \ref{effect of costs}.

\bibliographystyle{abbrv}

\end{CJK}
\end{document}